\documentclass[12pt,draftcls,peerreviewca,onecolumn,a4paper,dvips]{IEEEtran}
\usepackage{tipa}
\usepackage{amsfonts}
\usepackage{amsfonts}
\usepackage{amsfonts}
\usepackage{amsfonts}
\usepackage{amssymb}
\usepackage{stfloats}
\usepackage{cite}
\usepackage{graphicx}
\usepackage{psfrag}
\usepackage{subfigure}
\usepackage{amsmath}
\usepackage{array}
\usepackage[usenames]{color}

\DeclareMathOperator*{\argmin}{\arg\min}

\newcommand{\blue}{\color{black}}
\newcommand{\black}{\color{black}}
\newcommand{\red}{\color{black}}

\begin{document}

\title{Dynamic Partial Cooperative MIMO System for Delay-Sensitive
Applications with Limited Backhaul Capacity}

\newtheorem{Thm}{Theorem}
\newtheorem{Lem}{Lemma}
\newtheorem{Cor}{Corollary}
\newtheorem{Def}{Definition}
\newtheorem{Exam}{Example}
\newtheorem{Alg}{Algorithm}
\newtheorem{Prob}{Problem}
\newtheorem{Rem}{Remark}
\newtheorem{Proof}{Proof}
\newtheorem{Subproblem}{Subproblem}
\newtheorem{assumption}{Assumption}

%\author{\authorblockN{Ying Cui,  Vincent K. N. Lau and Huang Huang}
%\thanks{Ying Cui is with the Department of ECE, Northeastern University, Boston, MA, USA.  Vincent K. N. Lau is with the Department of ECE, Hong Kong University of Science and Technology (HKUST), Hong
%Kong. Huang Huang is with Huawei Technologies Co., Ltd., Shenzhen, P.R.China.}
%}

\author{\authorblockN{Ying Cui, \IEEEmembership{MIEEE}}
\authorblockA{Department of ECE\\ Northeastern University, USA} \and \authorblockN{Vincent K.N. Lau, \IEEEmembership{FIEEE}} \authorblockA{ Department of ECE \\ HKUST, Hong
Kong}\and\authorblockN{Huang Huang, \IEEEmembership{MIEEE}}\authorblockA{Huawei Technologies Co., Ltd.\\ Shenzhen, China}}
\maketitle

\begin{abstract}
\textcolor{black}{Considering backhaul consumption in practical systems, it may not be the best choice to engage all the time in full cooperative MIMO  for interference mitigation.} 
In this paper, we  propose a novel downlink partial cooperative MIMO
(Pco-MIMO) physical layer   (PHY) scheme, which allows flexible tradeoff between
the {\em partial data cooperation} level and the {\em backhaul
consumption}. Based on this Pco-MIMO scheme,  we consider  \textcolor{black}{dynamic} transmit power and  rate  allocation according to the \textcolor{black}{imperfect} channel state information at transmitters (CSIT) and  the queue state information (QSI) to minimize the average delay \textcolor{black}{cost} 
subject to \textcolor{black}{average} backhaul consumption constraints and average power
constraints. The delay-optimal control problem is formulated as an
infinite horizon average cost constrained partially observed Markov
decision process (CPOMDP). \textcolor{black}{By} exploiting the special structure in our
problem, we  derive an {\em equivalent Bellman Equation} to solve
the CPOMDP. \textcolor{black}{To reduce computational complexity} and facilitate distributed implementation, we propose
a distributed online learning algorithm to estimate the per-flow
potential functions and Lagrange multipliers (LMs) and
a distributed online stochastic partial  gradient algorithm to
obtain the power and rate control policy.  The proposed \textcolor{black}{low-complexity distributed} solution is 
based on local observations of the system states at the BSs and is very robust against model variations.  We also  prove the convergence and the asymptotic optimality of the proposed solution.
\end{abstract}


\newpage

%\IEEEpeerreviewmaketitle
\section{Introduction}\label{sec:intro}

%There are {\red many} works focusing on interference mitigation
%techniques for \textcolor{black}{downlink} wireless systems. According to the
%backhaul consumption  \textcolor{black}{requirement}, these   \textcolor{black}{techniques} can be roughly
%classified into two types, namely {\em coordinative MIMO} techniques
%and {\em cooperative MIMO} techniques.  For the coordinative MIMO
% \cite{IA:conventional:2008,IA:M*N_MIMO:2008} , only the  channel state information (CSI)  \textcolor{black}{(not the payload data)} is shared among  the base stations (BSs) for the beamformer design at  each BS to combat interference. Hence, the backhaul consumption   \textcolor{black}{of the coordinative MIMO techniques is relatively}  small at
%the expense of performance, e.g., degrees of freedom (DoF). On the other hand, for the
%cooperative MIMO  \cite{Gesbert:jsac:2010,zhang:coor:2004}, both the CSI and the payload data are shared among the BSs for the joint precoder designs at the BSs  \textcolor{black}{to  increase the
%spectral efficiency dramatically}. However, the  \textcolor{black}{significant}
%performance  \textcolor{black}{gain} of \textcolor{black}{the} cooperative MIMO techniques \textcolor{black}{comes} at the expense of
%\textcolor{black}{increased} backhaul consumption to deliver the shared payload data between
%the BSs.

\subsection{Background}

There are many works focusing on interference mitigation techniques for downlink wireless systems. According to the backhaul consumption requirement, these techniques can be roughly classified into two types, namely, coordinative MIMO techniques and cooperative MIMO techniques. \textcolor{black}{For the coordinative MIMO \cite{DahroujYu:2010,Gesbertadaptation:2007,RenLaDLbf:2006}, each base station (BS) serves a disjoint set of mobile users (MSs), but designs its beamformer jointly with all other BSs to reduce inter-cell interference. Therefore, only the channel state information (CSI) (not the payload data) is shared among BSs for the beamformer design at each BS to combat interference. The backhaul consumption of the coordinative MIMO techniques is relatively small at the cost of performance, e.g., degrees of freedom (DoF). On the other hand, for the cooperative MIMO \cite{FoschiniNetworkMIMO:2006, ShamaicoMIMO:2001,zhang:coor:2004}, all BSs serve and coordinate interference to all MSs. Therefore, both the CSI and the payload data are shared among the BSs and the network becomes a broadcast channel topology with joint precoding at the BSs. However, the significant performance gain of the cooperative MIMO techniques \textcolor{black}{comes} at the cost of increased backhaul consumption to deliver the shared payload data among the BSs. (See \cite{Gesbert:jsac:2010}, \cite{BjornsonEduardMulticellMIMO:2012} and references therein for surveys of recent results on coordinative and cooperative MIMO.)}

It is obvious that
when backhaul  constraints are  imposed, \textcolor{black}{it may not be optimal to always engage in full MIMO cooperation  to mitigate interference}.
 Recently, there have been some research works on  partial MIMO cooperation. For example, \textcolor{black}{in \cite{Marsch08onmulti-cell}, the authors considered joint user selection, antenna selection and power control for backhaul constrained downlink cooperative transmission in a multi-cell network where each BS has multiple antennas, while each MS has one  antenna.  MIMO cooperation is only done among the selected antennas. A heuristic solution adaptive to the CSI was proposed.} In
\cite{MIMO:int:huang}, the authors proposed a uni-directional \textcolor{black}{MIMO cooperation design (called Uco-MIMO here)} \textcolor{black}{for a two multi-antenna transmitter and two multi-antenna receiver setup} to reduce the backhaul consumption. However, the design is static in the sense
that it always  engages in the same uni-directional data
sharing (consuming the same backhaul capacity) in the entire
communication session and fails to  capture good channel opportunities in \textcolor{black}{dynamic wireless systems}. In \cite{RandapartialcoMIMO:2011} and \cite{SiddarthpartialcoMIMO:2010}, the authors proposed partial MIMO cooperation designs \textcolor{black}{for a two multi-antenna transmitter and two single-antenna receiver setup} based on common-private rate splitting schemes under backhaul constraints. The rate splitting schemes are adaptive to the CSI only.

\textcolor{black}{In this paper,
we are
interested in designing a novel downlink partial cooperative MIMO
(Pco-MIMO) physical layer   (PHY) scheme, which allows flexible tradeoff between
the {\em partial data cooperation} level and the {\em backhaul
consumption}. Based on this Pco-MIMO scheme,  we consider  \textcolor{black}{dynamic} transmit power and  rate  allocation according to the channel state information at transmitters (CSIT) and  the queue state information (QSI) to minimize the average delay \textcolor{black}{cost} 
subject to \textcolor{black}{average} backhaul consumption constraints and average power
constraints. The motivations and challenges of this work are summarized below.}

$\bullet$ {\bf Flexible Partial Cooperative MIMO PHY Scheme:} \textcolor{black}{The existing partial cooperative MIMO designs  in \cite{Marsch08onmulti-cell,MIMO:int:huang,RandapartialcoMIMO:2011, SiddarthpartialcoMIMO:2010} have certain restrictions on the  cooperation level (e.g., MIMO cooperation among selected antennas in \cite{Marsch08onmulti-cell} and uni-directional MIMO cooperation in \cite{MIMO:int:huang}) or  the network configuration (e.g., the two multi-antenna
transmitter and two single-antenna receiver configuration in \cite{RandapartialcoMIMO:2011} and \cite{SiddarthpartialcoMIMO:2010}).} It is \textcolor{black}{quite challenging} 
to design a  PHY  scheme that supports \textcolor{black}{flexible adjustment of the cooperation level (embracing the full coordinative MIMO, partial cooperative MIMO and full cooperative MIMO  schemes as special cases)  as illustrated in Fig. \ref{fig:coor_coop}.  Furthermore, the scheme should also be applicable to a general multi-BS multi-antenna configuration.}

\begin{figure}[h]
 \begin{center}
\includegraphics[height=1.4cm, width=8cm]{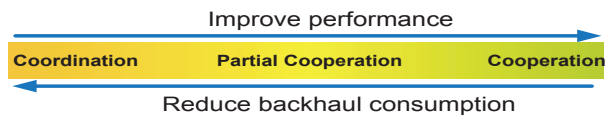}
 \end{center}
    \caption{Illustration of the family of {\em interference mitigation techniques}.}
    \label{fig:coor_coop}
\end{figure}

$\bullet$  {\bf Delay-Aware Dynamic Partial Cooperative MIMO Control:}
\textcolor{black}{The existing resource control designs for the partial cooperative MIMO schemes  in \cite{Marsch08onmulti-cell} and \cite{RandapartialcoMIMO:2011} are adaptive to the CSI only. A common assumption in these existing works is that there  are infinite backlogs at the
transmitters and the applications are delay-insensitive.  However, in practice, a lot of applications have bursty data arrivals and they are
delay-sensitive. Therefore, it is very important to take into account the delay performance in designing 
partial cooperative MIMO schemes.}
To support delay-sensitive
applications,
the dynamic resource control  should be jointly adaptive to the CSI
and the QSI   in the system \textcolor{black}{to exploit the 
information regarding the {\em transmission opportunity} (provided by the CSI) and 
the {\em data urgency} (provided by the QSI).}  Striking an optimal balance between the \textcolor{black}{transmission}
opportunity  and the data urgency for delay-sensitive applications is very
{\red challenging},  \textcolor{black}{because it involves solving an infinite horizon stochastic optimization problem \cite[Chap. 4]{Bertsekas:2007}.} \textcolor{black}{The brute force solutions using stochastic optimization techniques \cite[Chap. 4]{Bertsekas:2007} lead to centralized delay-optimal control policies. These brute force solutions have exponential complexity with respect to (w.r.t.) the number of data streams and  require the global QSI of all data streams. Therefore, it is highly desirable and challenging to obtain a low-complexity distributed delay-aware  resource control design for practical  multi-cell MIMO networks.}

$\bullet$  {\bf Impact of Imperfect CSIT:} \textcolor{black}{The resource control designs for the partial cooperative MIMO schemes  in \cite{Marsch08onmulti-cell} and \cite{RandapartialcoMIMO:2011} assume perfect CSIT.}
In practice, the CSIT may  be imperfect due to \textcolor{black}{the duplexing delay  in TDD systems \cite{TDD:2006} or feedback latency and quantization in FDD systems \cite{FDD:2007}}.
With  imperfect CSIT,  there may be packet errors in each frame due to the
uncertainty of the mutual information  at the
transmitters. \textcolor{black}{Thus, it is important to take into account  the CSIT errors in the  resource optimization design. Yet, this requires explicit knowledge of the statistics of the CSIT errors and the bursty data arrivals. It is quite challenging to have a robust solution w.r.t. uncertainty in the modeling.}

\subsection{Main Contributions}

\textcolor{black}{In this paper, we consider a general multi-BS multi-antenna MIMO network, the system model of which is presented in  Section \ref{sec:model}. In Section \ref{sec:partMIMO},} we propose a novel Pco-MIMO PHY scheme,  \textcolor{black}{which allows flexible tradeoff between the MIMO cooperation level and the backhaul consumption}.  \textcolor{black}{In Section \ref{sec:delay-optimal-form},  we formulate the delay-optimal transmit
power and  rate allocation  (according to the imperfect 
CSIT and  the QSI)  under  average power and backhaul consumption constraints as an infinite horizon
average cost constrained partially observed Markov decision process
(CPOMDP) \cite[Chap. 4]{Bertsekas:2007}, \cite{Borkaractorcritic:2005}, \cite{Meuleau:1999}.}   By exploiting the special structure in our problem, we
derive an  equivalent \textcolor{black}{optimality equation} to solve the CPOMDP in Section \ref{sec:opt_solution}. 
\textcolor{black}{In Section \ref{sec:low-complexity-solution}, we propose a  distributed online power and rate control solution using a distributed stochastic learning algorithm and a distributed online stochastic partial gradient algorithm.
The proposed solution has very low complexity and requires only 
 local observations of the system states at the BSs. Hence, it can be implemented distributively and is very robust against model variations. We also establish technical conditions for the convergence and asymptotic optimality of the proposed solution.
We demonstrate the significant performance gain of the proposed scheme
compared with various baseline schemes using numerical simulations  in Section
\ref{sec:simu}. Finally, conclusions are provided in Section
\ref{sec:summary}.}

\subsection{Notation}

$\mathbb{C}$ and $\mathbb{R}$ denote the sets of complex and real
numbers, respectively. $\mathbb{E}[\cdot]$ and $\mathbf{1}(\cdot)$
denote expectation and the indicator function, respectively.
$|\cdot|$ denotes the absolute value function for a scalar or the
cardinality  for a set. $(\cdot)^T$ and $(\cdot)^\dag$
denote the transpose and Hermitian transpose, respectively.
$\lfloor\cdot\rfloor$ denotes the floor function. $[x]^+=\max(x,0)$
and $[x]_{\bigwedge N_Q}=\min(x,N_Q)$. $\text{diag}(a_1,\cdots,a_n)$
denotes a diagonal matrix with diagonal entries $a_1,\cdots,a_n$.
$\text{Null}(\mathbf{H})$ denotes the null space of matrix
$\mathbf{H}$. $[\mathbf{H}]_{(l,m)}$ denotes the $(l,m)$-th
element of $\mathbf{H}$.

\section{System Model}\label{sec:model}

\subsection{System Topology}

\textcolor{black}{We consider the downlink transmission of a MIMO network  with $K$ multi-antenna
BSs delivering $K$ delay-sensitive data flows to $K$ multi-antenna
MSs, where $K\geq 2$.   Fig. \ref{fig:system_model} illustrates an example with $K=2$.} Specifically, each BS is
equipped with $M$ antennas, while each MS is equipped with \textcolor{black}{$N$ antennas, where $\frac{M}{K}< N\leq M$} 
 \footnote{If \textcolor{black}{$\frac{M}{K}\geq N$}, we could simply  apply \textcolor{black}{coordinative MIMO} to achieve the maximum DoF of $KN$ without any
data cooperation. Hence, we do not consider  this trivial case.}.
Furthermore, we consider the e-NodeB architecture in LTE systems
\cite{LTE:2008}, where  \textcolor{black}{every} two BSs are connected by a bi-directional
backhaul link with limited capacity  in each direction. \textcolor{black}{Denote  $\mathcal K\triangleq\{1,2,\cdots, K\}$.  Each BS $k$ has one MS (indexed by $k$) in its cell and maintains a queue for the bursty
data flow towards MS $k$, where $k\in \mathcal K$. BS $k$ is the master BS for serving MS $k$, while the other BSs $n\in \mathcal K, n\neq k$ cooperatively serve MS $k$ according to the proposed Pco-MIMO scheme, which will be illustrated in Section \ref{sec:pco-mimo}.} The time
dimension is partitioned into scheduling frames indexed by $t$  with  frame duration
$\tau$ (seconds).

\begin{figure}[h]
 \begin{center}
\includegraphics[height=3.5cm, width=12cm]{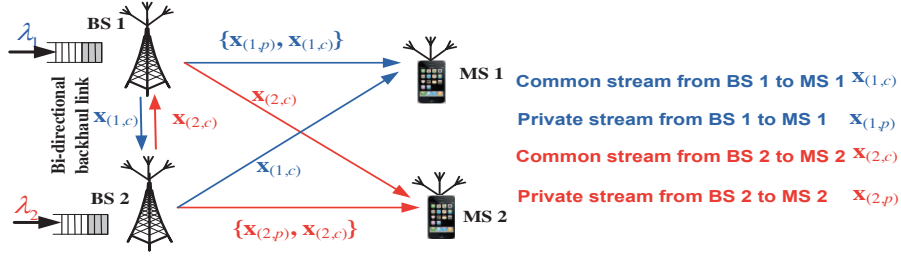}
 \end{center}
    \caption{
\scriptsize{System Model for $K=2$. There are two multi-antenna BSs connected by a
bi-directional backhaul link with limited backhaul capacity. The
flows into BS 1 and BS 2 are targeted to MS 1 and MS 2,
respectively. \textcolor{black}{The data streams to each MS are split into 
{\em common streams} and  {\em private streams}.  The common streams
are shared  through the backhaul and jointly transmitted by the two BSs. The private streams are
transmitted locally at  each BS to its target MS. }}}
    \label{fig:system_model}
\end{figure}

\subsection{MIMO Channel and Imperfect CSIT Models}
Let $\mathbf{s}_n\in\mathbb{C}^{M\times 1}$ be the complex
signal vector transmitted by BS  $n$ and
$\mathbf{z}_k\in\mathbb{C}^{N\times 1}$ be the circularly symmetric
Additive White Gaussian Noise (AWGN) vector at MS $k$, \textcolor{black}{where $n,k\in \mathcal K$}. We assume all
noise terms are i.i.d. zero mean complex Gaussian with
$\mathbb{E}[\mathbf{z}_k(\mathbf{z}_k)^{H}]=\mathbf{I}_{N}$.   The MIMO channel output
$\mathbf{y}_k\in\mathbb{C}^{N\times 1}$ \textcolor{black}{at MS $k\in \mathcal K$} is given by:
\begin{equation}\label{eq:system_model}
\mathbf{y}_k = \sum\nolimits_{n\in\mathcal{K}}
\mathbf{H}_{kn}\textcolor{black}{ \mathbf{s}_n} + \mathbf{z}_k,
\end{equation}
where $\mathbf{H}_{kn}\in\mathcal H^{N\times
M}$ is the MIMO complex fading coefficient (CSI) from BS $n$ to MS $k$ and $\mathcal H\subset \mathbb C$ denotes the finite discrete  CSI state space.
Let $\mathbf{H}_{kn}(t)$ denote the CSI from BS $n$ to MS $k$ at frame $t$. We have the following assumption on the CSI.
\begin{assumption}[CSI Model]\label{ass:channel_model}
The $(i,l)$-th element $[\mathbf{H}_{kn}]_{(i,l)}(t)$  is constant within each frame and  i.i.d. over scheduling frame $t$ following a general distribution  over $\mathcal H$. $\{[\mathbf{H}_{kn}]_{(i,l)}(t)\}$ is
 independent w.r.t. $\{k,n,i,l\}$. Assume $\text{rank}(\mathbf{H}_{kn}(t)) = \min(M,N)$ with probability 1. The BSs  do not have knowledge of the CSI distribution $\Pr\{\mathbf{H}_{kn}\}$.~\hfill\IEEEQED
\end{assumption}

We
assume that each MS has perfect knowledge of the CSI (perfect CSIR)  but
each BS only has imperfect knowledge of the CSI (imperfect CSIT)\footnote{\textcolor{black}{Most of  CSIR estimation errors come from the pilot/preamble estimation noise at the receiver. In practical systems, such as LTE and Wimax, the pilot power is designed to be sufficient for CSIR estimations at the receiver to support the demodulation of 64QAM. On the other hand, CSIT errors come from duplexing delay in TDD systems or feedback latency/quantization in FDD systems. Hence, they are usually much larger than CSIR errors. As a result, we consider perfect CSIR,  but imperfect CSIT  \cite{TDD:2006},\cite{FDD:2007}.}}.   \textcolor{black}{Thus, $\mathbf{H}_{kn}\in\mathcal H^{N\times
M}$ also denotes the (accurate) MIMO complex fading coefficient  from BS $n$ to MS $k$ estimated at MS $k$.} Let $\hat{\mathbf H}_{kn}\in \hat{\mathcal H}^{N\times
M}$ denote the imperfect MIMO complex fading coefficient from BS $n$ to MS $k$ estimated \textcolor{black}{(with error)} at BS $n$, where $\hat{\mathcal H}\subset \mathbb C$ denotes the  finite discrete CSIT state space. 
\begin{assumption}[Imperfect CSIT Model]\label{ass:csit_model}
The imperfect CSIT $\mathbf{\hat{H}}_{kn}$ is stochastically related to
the actual CSI $\mathbf{H}_{kn}$ via  the CSIT error kernel
$\Pr\{\mathbf{\hat{H}}_{kn}|\mathbf{H}_{kn}\}$. Assume
$\text{rank}(\hat{\mathbf{H}}_{kn}) = \min(M,N)$ with probability 1.
The BSs \textcolor{black}{do not have knowledge of the} CSIT error kernel.~\hfill\IEEEQED
\end{assumption}
The imperfect CSIT model in Assumption \ref{ass:csit_model} is very
general and  covers  most of the cases we encounter in practice, e.g., \textcolor{black}{the imperfect CSIT due to duplexing delay in TDD systems \cite{TDD:2006} or feedback latency and quantization in FDD systems \cite{FDD:2007}.} We denote \textcolor{black}{$\mathbf{H}=\{\mathbf{H}_{kn}:k,n\in \mathcal K\}$ and $\hat{\mathbf H}=\{\hat{\mathbf H}_{kn}: k,n\in \mathcal K\}$} as the {\em global CSI} and the {\em global CSIT}, respectively.

\subsection{Bursty Source Model and Queue Dynamics}

Let \textcolor{black}{$\mathbf{A}(t) = \{A_k(t):k\in\mathcal K\} $} be the random new arrivals
(number of bits) to the $K$ BSs at the end of
frame \textcolor{black}{$t$}.
\begin{assumption}[Bursty Source Model]
The arrival  $A_k(t)$ is i.i.d. over scheduling frame \textcolor{black}{$t$} and
\textcolor{black}{follows} a general distribution. The
average arrival rate is $\lambda_{k}=\mathbb{E}[A_k]$. Furthermore, the
random arrival process $\{A_k\textcolor{black}{(t)}\}$ is independent w.r.t. $k$.
~\hfill\IEEEQED
\end{assumption}

Let \textcolor{black}{$\mathbf{Q}(t) =\{Q_k(t): k\in \mathcal K\}\in\boldsymbol{\mathcal{Q}}$} denote the
global \textcolor{black}{QSI}  at the
beginning of frame $t$, where $\boldsymbol{\mathcal{Q}}$ is the state space for
the global QSI. $N_Q$ denotes the buffer size (number of bits).  
The
queue dynamics of \textcolor{black}{MS $k\in \mathcal K$} is given by:
\begin{equation}\label{eq:Q_org}\begin{array}{lll}\blue
Q_{k}(t+1) = \Big[\big[Q_k(t)-U_k(t)\big]^+ + A_k(t)
\Big]_{\bigwedge N_Q},
\end{array}
\end{equation}
where  $U_k(t)$  is the goodput (number
of bits successfully received) at  MS $k$ \textcolor{black}{at the end of frame $t$}. The expression of  $U_k(t)$ will be given in \eqref{eq:Ubit_k} .

\subsection{Power Consumption Model}

At frame $t$, the total power consumption $P_k(t)$ of BS $k\in \mathcal K$ is contributed by  the constant circuit power
of the RF chains (such as the mixers, synthesizers and
digital-to-analog converters)  $P_{cct}$ and the
transmit power of the power amplifier (PA) $P_k^{tx}(t)$ as follows:
\begin{equation}\label{eq:p_k}
P_k(t) = P_k^{tx}(t) + P_{\text{cct}}\mathbf{1}(P_k^{tx}(t)>0).
\end{equation}
$P_{\text{cct}}$ is constant irrespective of $P_k^{tx}(t)$. \textcolor{black}{The expression of $P_k^{tx}(t)$ will be given in \eqref{eq:p_tx}.}

\section{Partial Cooperative MIMO and DoF Analysis}\label{sec:partMIMO}

In this section, we first propose a novel {\em Partial Cooperative MIMO}
(Pco-MIMO) \textcolor{black}{PHY} scheme. \textcolor{black}{Then, we analyze the associated system DoF performance.}

\subsection{Partial Cooperative MIMO Scheme}\label{sec:pco-mimo}
\textcolor{black}{As illustrated in Fig. \ref{fig:system_model},} the data streams to each MS $k\in \mathcal K$ are split into $d_{(k,c)}\in \mathbb N$
{\em common streams} and $d_{(k,p)}\in \mathbb N$ {\em private streams}, \textcolor{black}{where $\mathbb N$ denotes the set of natural numbers}.  The common streams
are shared  through the backhaul and jointly transmitted by the $K$ BSs.  As a result, some backhaul
capacity is consumed. On the other hand, the private streams are
transmitted   locally at  each BS and no backhaul consumption is
incurred. \textcolor{black}{We adopt zero-forcing (ZF) precoder and decorrelator designs at the BSs and MSs, respectively.  To fully eliminate interference and recover $d_{(k,c)}$ common streams and $d_{(k,p)}$ private streams at each MS $k$ when the CSIT is perfect, we have some conditions on  $d_{(k,c)}$ and  $d_{(k,p)}$ for all $k\in \mathcal K$.  First, to transmit common streams using MIMO cooperation at the \textcolor{black}{$K$} BSs, we require $d_{(k,c)}\leq \min(KM,N)=N$. Next, $d_{K,M,N}\triangleq[M-(K-1)N]^+$ private streams to MS $k$ can be zero-forced  at  BS $k$ to eliminate interference at MS $n\in \mathcal K, n\neq k$. Thus, we can choose $d_{(k,p)}$ satisfying $ d_{K,M,N}\leq d_{(k,p)}\leq \min(M,N)=N$. (Note that this condition is valid as the assumption $\frac{M}{K}<N$ implies $d_{K,M,N}<N$.)  Finally, for each MS $k$ to  eliminate the residual interference from the remaining $d_{n,p}-d_{K,M,N}$ private streams to MS $n$ using ZF decorrelation and detect all the desired streams, we require  $d_{(k,p)}+d_{(k,c)}+\sum_{n\in \mathcal K,  n\neq k}\left(d_{(n,p)}-d_{K,M,N}\right) \leq
N$. Therefore, we have the following feasibility constraints on $\{d_{(k,c)},d_{(k,p)}:k\in \mathcal K\}$: \footnote{\textcolor{black}{Note that $d_{(k,p)}\geq  d_{K,M,N}$ and $d_{(k,p)}+d_{(k,c)}+\sum_{n\in \mathcal K, n\neq k}\left(d_{(n,p)}-d_{K,M,N}\right) \leq
N$ imply  $d_{(k,c)}\leq N$ and $d_{(k,p)}\leq N$.}}
\begin{align}\label{eq:d_k}
d_{(k,p)}\geq  d_{K,M,N}, \ d_{(k,c)}+\sum_{n\in \mathcal K}d_{(n,p)} \leq
N+(K-1)d_{K,M,N}, \ \forall k \in \mathcal K.
\end{align}
Note that \eqref{eq:d_k} implies 
\begin{align}
d_{(k,c)}\leq N- d_{K,M,N}, \ \forall k\in \mathcal K.\label{eqn:d_kc}
\end{align}}

Let $x_{(k,c)}^i$ denote the $i$-th common stream transmitted
from the \textcolor{black}{$K$} BSs to MS $k$, where $i\in\{1,\cdots,d_{(k,c)}\}$. Let
$x_{(k,p)}^i$ denote the $i$-th private stream transmitted from
BS $k$ to MS $k$, where $i\in\{1,\cdots,d_{(k,p)}\}$. Furthermore,
let $P_{(k,c)}^i$ and $P_{(k,p)}^i$ denote the \textcolor{black}{transmit power for}
 $x_{(k,c)}^i$ and $x_{(k,p)}^i$,
respectively. Let $\mathbf{b}_{(k,c)}^i\in\mathbb{C}^{\textcolor{black}{K}M\times1}$ 
denote the joint precoder for $x_{(k,c)}^i$ at the \textcolor{black}{$K$} BSs, where $|\mathbf{b}_{(k,c)}^i|=1$.  
Let $\mathbf{b}_{(k,p)}^i\in\mathbb{C}^{M\times1}$ denotes the precoder for
$x_{(k,p)}^i$  at BS $k$, where $|\mathbf{b}_{(k,p)}^i|=1$. \textcolor{black}{Let $\mathbf s_k\in\mathbb{C}^{M\times1}$ be the complex signal vector transmitted by BS $k$.} Then, the complex signal vector transmitted by
the  \textcolor{black}{$K$} BSs is given by: 
\textcolor{black}{
\begin{align}
\begin{bmatrix}
\mathbf s_1\\ \vdots  \\ \mathbf s_K
\end{bmatrix}=\sum_{k\in \mathcal K}\left(\mathbf{B}_{(k,c)}\Sigma_{(k,c)}\mathbf{x}_{(k,c)}\right)+
\begin{bmatrix}
\mathbf{B}_{(1,p)}\Sigma_{(1,p)}\mathbf{x}_{(1,p)}\\ \vdots  \\ \mathbf{B}_{(K,p)}\Sigma_{(K,p)}\mathbf{x}_{(K,p)}
\end{bmatrix},\label{eqn:s-x}
\end{align}}
%where \textcolor{black}{for all $k\in \mathcal K$, we have}
$$\text{where }\mathbf{B}_{(k,c)}=[\mathbf{b}_{(k,c)}^1,\cdots,\mathbf{b}_{(k,c)}^{d_{(k,c)}}]\in\mathbb{C}^{\textcolor{black}{K}M\times{d_{(k,c)}}}, \quad \mathbf{B}_{(k,p)}=[\mathbf{b}_{(k,p)}^1,\cdots,\mathbf{b}_{(k,p)}^{d_{(k,p)}}]\in\mathbb{C}^{M\times{d_{(k,p)}}},$$
$$\Sigma_{(k,c)}=\text{diag}(\sqrt{P_{(k,c)}^1},\cdots,
\sqrt{P_{(k,c)}^{d_{(k,c)}}}),\quad \Sigma_{(k,p)}=\text{diag}(\sqrt{P_{(k,p)}^1},\cdots,
\sqrt{P_{(k,p)}^{d_{(k,p)}}})$$
$$\mathbf{x}_{(k,c)}=[x_{(k,c)}^1,\cdots,x_{(k,c)}^{d_{(k,c)}}]^T,\quad \mathbf{x}_{(k,p)}=[x_{(k,p)}^1,\cdots,x_{(k,p)}^{d_{(k,p)}}]^T.$$
Substituting \eqref{eqn:s-x} into \eqref{eq:system_model}, 
the received signal \textcolor{black}{$\mathbf{y}_k\in \mathbb C^{N\times1}$} at each MS $k\in \mathcal K$ is given by:
\begin{equation}\label{eq:y_k}\begin{array}{lll}
&\mathbf{y}_k=
\underbrace{
\mathbf{H}_{k}\mathbf{B}_{(k,c)}\Sigma_{(k,c)}\mathbf{x}_{(k,c)}+\mathbf{H}_{kk}\mathbf{B}_{(k,p)}\Sigma_{(k,p)}\mathbf{x}_{(k,p)}}_{\text{desired
signals for MS $k$}}  +\underbrace{\textcolor{black}{\sum_{n\in \mathcal K, n\neq k}}
\mathbf{H}_{k}\mathbf{B}_{(n,c)}\Sigma_{(n,c)}\mathbf{x}_{(n,c)}}_{\text{interference
from  common streams \textcolor{black}{to MS} $n$}}
\\
&+\underbrace{\textcolor{black}{\sum_{n\in \mathcal K, n\neq k}}\mathbf{H}_{kn}\mathbf{B}_{(n,p)}^{(1)}\Sigma_{(n,p)}^{(1)}\mathbf{x}_{(n,p)}^{(1)}}_{\text{interference
from  first $\textcolor{black}{d_{K,M,N}}$ private streams \textcolor{black}{to MS} $n$}}+\underbrace{\textcolor{black}{\sum_{n\in \mathcal K, n\neq k}}\mathbf{H}_{kn}\mathbf{B}_{(n,p)}^{(2)}\Sigma_{(n,p)}^{(2)}\mathbf{x}_{(n,p)}^{(2)}}_{\text{interference
from   last $d_{(n,p)}-\textcolor{black}{d_{K,M,N}}$ private streams \textcolor{black}{to MS} $n$}}\\
&+ \mathbf{z}_k, 
\end{array}
\end{equation}
where \textcolor{black}{$\mathbf{H}_{k}=\left[\mathbf{H}_{k1},\cdots,
\mathbf{H}_{kK}\right]\in\mathbb{C}^{N\times KM}$} and \textcolor{black}{
$$\mathbf{B}_{(n,p)}^{(1)}=[\mathbf{b}_{(n,p)}^1,\cdots,\mathbf{b}_{(n,p)}^{d_{K,M,N}}]\in\mathbb{C}^{M\times d_{K,M,N}}, \ \mathbf{B}_{(n,p)}^{(2)}=[\mathbf{b}_{(n,p)}^{d_{K,M,N}+1},\cdots,\mathbf{b}_{(n,p)}^{d_{(n,p)}}] \in\mathbb{C}^{M\times \left(d_{(k,p)}-d_{K,M,N}\right)},$$
$$\Sigma_{(k,p)}^{(1)}=\text{diag}(\sqrt{P_{(k,p)}^1},\cdots,
\sqrt{P_{(k,p)}^{d_{K,M,N}})},\ \Sigma_{(k,p)}^{(2)}=\text{diag}(\sqrt{P_{(k,p)}^{d_{K,M,N}+1}},\cdots,
\sqrt{P_{(k,p)}^{d_{(k,p)}})},$$
$$\mathbf{x}_{(n,p)}^{(1)}=[x_{(n,p)}^1,\cdots,x_{(n,p)}^{d_{K,M,N}}]^T, \ \mathbf{x}_{(n,p)}^{(2)}=[x_{(n,p)}^{d_{K,M,N}+1},\cdots,x_{(n,p)}^{d_{(n,p)}}]^T.$$ Note that 
$\mathbf{B}_{(n,p)}=[\mathbf{B}_{(n,p)}^{(1)},\mathbf{B}_{(n,p)}^{(2)}]$,
$\mathbf{\Sigma}_{(n,p)}=\left[\mathbf{\Sigma}_{(n,p)}^{(1)},0\atop
0,\mathbf{\Sigma}_{(n,p)}^{(2)}\right]$, and
$\mathbf{x}_{(n,p)}=[(\mathbf{x}_{(n,p)}^{(1)})^T,(\mathbf{x}_{(n,p)}^{(2)})^T]^T$.}\footnote{\textcolor{black}{Without loss of generality, we present the precoder design for the case where $M>N$. When $M=N$ (i.e., $d_{K,M,N}=0$), we can directly adopt the precoder design for the last  $d_{(n,p)}-d_{K,M,N}$ private streams in the case where $M>N$.}}
\textcolor{black}{Let $\mathbf{u}_{(k,p)}^i\in\mathbb{C}^{N\times1}$ and
$\mathbf{u}_{(k,c)}^i\in\mathbb{C}^{N\times1}$ be the decorrelators
for $x_{(k,p)}^i$ and $x_{(k,c)}^i$, respectively. After decorrelation, the recovered signals $r_{(k,c)}^i$ and $r_{(k,p)}^i$ for   $x_{(k,c)}^i$ and  $x_{(k,p)}^i$ at MS $k$ are:
\begin{align}
r_{(k,c)}^i=(\mathbf{u}_{(k,c)}^i)^{\dag}\mathbf{y}_k, \quad r_{(k,p)}^i=(\mathbf{u}_{(k,p)}^i)^{\dag}\mathbf{y}_k.\label{eqn:recover-r}
\end{align}}

In the following, we present the \textcolor{black}{precoder and
decorrelator} designs for the Pco-MIMO under perfect CSIT, \textcolor{black}{i.e., $\hat{\mathbf H}=\mathbf H$}. \textcolor{black}{Note that when the CSIT is imperfect, the precoders at the BSs are designed according to imperfect CSIT $\hat{\mathbf H}$ instead of $\mathbf H$. Thus, there will be residual
interference  from the common streams and the first $d_{K,M,N}$ private streams for other MSs even after the imperfect ZF
precoding.} The impact of the 
imperfect CSIT will be discussed in  \textcolor{black}{Section}
\ref{sec:imp_CSI}.

\subsubsection{Precoder Design for Pco-MIMO}\textcolor{black}{First, we design the precoder  $\mathbf{B}_{(n,c)}$ at the $K$ BSs for the common streams $\mathbf{x}_{(n,c)}$, where $ n\in \mathcal K$.}
To eliminate the interference term 
$\mathbf{H}_{k}\mathbf{B}_{(n,c)}\Sigma_{(n,c)}\mathbf{x}_{(n,c)}$ 
in \eqref{eq:y_k} experienced by MS $k\neq n$,  the joint ZF precoder   $\mathbf{B}_{(n,c)}\in\mathbb{C}^{\textcolor{black}{K}M\times{d_{(n,c)}}}$ at the \textcolor{black}{$K$} BSs is given
by \cite{weiyu:2010}:
\begin{equation}\label{eq:v_kc}\blue
\mathbf{B}_{(n,c)}=\widetilde{\mathbf{B}}_{\textcolor{black}{n}}\mathbf{F}_{(n,c)}, 
\end{equation} 
where the columns of
$\widetilde{\mathbf{B}}_n\in\mathbb{C}^{\textcolor{black}{K}M\times (\textcolor{black}{K}M-N)}$
form the orthonormal basis of $$\text{Null}(\textcolor{black}{\left[\mathbf{H}_{1}^T,\cdots,\mathbf{H}_{n-1}^T, \mathbf{H}_{n+1}^T,\cdots, \mathbf{H}_{K}^T\right]^T }).$$  
$\mathbf{F}_{(n,c)}=[\mathbf{\overline{b}}_{(n,c)}^{1},\cdots,\mathbf{\overline{b}}_{(n,c)}^{d_{(n,c)}}]\in\mathbb{C}^{(\textcolor{black}{K}M-N)\times d_{(\textcolor{black}{n},c)}}$ is
designed by performing SVD on
$\mathbf{H}_{n}\widetilde{\mathbf{B}}_{\textcolor{black}{n}}$
\cite{Spencer:Tsig2004}, i.e., $\mathbf{H}_{n}\widetilde{\mathbf{B}}_{\textcolor{black}{n}}
=\mathbf{\overline{U}}_{(n,c)}\overline{\Sigma}_{(n,c)}
[\mathbf{\overline{b}}_{(n,c)}^{1},\cdots,\mathbf{\overline{b}}_{(n,c)}^{\textcolor{black}{K}M-N}]^{\dag}$, 
%\begin{equation}\begin{array}{lll}\label{eq:c_common}\blue
%\mathbf{H}_{n}\widetilde{\mathbf{B}}_{\textcolor{black}{n}}
%=\mathbf{\overline{U}}_{(n,c)}\overline{\Sigma}_{(n,c)}
%[\mathbf{\overline{b}}_{(n,c)}^{1},\cdots,\mathbf{\overline{b}}_{(n,c)}^{\textcolor{black}{K}M-N}]^{\dag},
%\end{array}
%\end{equation}
where the eigenvalues in $\overline{\Sigma}_{(n,c)}$ are sorted in 
decreasing order along the diagonal. Therefore, the common streams $\mathbf{x}_{(n,c)}$
are transmitted on the dominant eigenmodes for the desired MS $n$.

\textcolor{black}{Next, we design the precoders $\mathbf{B}_{(n,p)}^{(1)}$ and 
$\mathbf{B}_{(n,p)}^{(2)} $ at BS $n$ for the first $\textcolor{black}{d_{K,M,N}}$ private streams $\mathbf{x}_{(n,p)}^{(1)}$ and the last $d_{(n,p)}-\textcolor{black}{d_{K,M,N}}$ private streams $\mathbf{x}_{(n,p)}^{(2)}$, respectively. To eliminate the interference term
$\mathbf{H}_{kn}\mathbf{B}_{(n,p)}^{(1)}\Sigma_{(n,p)}^{(1)}\mathbf{x}_{(n,p)}^{(1)}$  in \eqref{eq:y_k} experienced by MS $k$, the ZF precoder  $\mathbf{B}_{(n,p)}^{(1)}\in\mathbb{C}^{M\times \textcolor{black}{d_{K,M,N}}}$ at BS $n$ is given by:}
\begin{equation}\label{eq:v_kz}
\mathbf{B}_{(n,p)}^{(1)}=\widetilde{\mathbf{B}}_{kn}\mathbf{F}_{(n,\textcolor{black}{p})},
\end{equation}
where the columns of
$\widetilde{\mathbf{B}}_{kn}\in\mathbb{C}^{M\times \textcolor{black}{d_{K,M,N}}}$
form the orthonormal basis of $\text{Null}(\mathbf{H}_{kn})$ and
$\mathbf{F}_{(n,\textcolor{black}{p})}\in\mathbb{C}^{\textcolor{black}{d_{K,M,N}}\times \textcolor{black}{d_{K,M,N}}}$ is designed by
performing SVD on
$\mathbf{H}_{nn}\widetilde{\mathbf{B}}_{kn}$, i.e.,
$\mathbf{H}_{nn}\widetilde{\mathbf{B}}_{kn}
=\mathbf{\overline{U}}_{(n,\textcolor{black}{p})}\overline{\Sigma}_{(n,z)}
(\mathbf{F}_{(n,\textcolor{black}{p})})^{\dag}.$ The eigenvalues in
$\overline{\Sigma}_{(n,\textcolor{black}{p})}$ are sorted in decreasing order along
the diagonal. \textcolor{black}{The precoder $\mathbf{B}_{(n,p)}^{(2)} \in\mathbb{C}^{M\times \left(d_{(n,p)}-\textcolor{black}{d_{K,M,N}}\right)}$ at BS $n$ 
is chosen to maximize the SNR of for the remaining $d_{(n,p)}-\textcolor{black}{d_{K,M,N}}$ private streams
$\mathbf{x}_{(n,p)}^{(2)}$  of the  MS $n$, i.e.,}
\begin{equation}\label{eq:v_ks}
\mathbf{B}_{(n,p)}^{(2)} =
[\mathbf{\overline{b}}_{nn}^1,\cdots,\mathbf{\overline{b}}_{nn}^{d_{(n,p)}-\textcolor{black}{d_{K,M,N}}}]
\end{equation} is obtained by
performing SVD on $\mathbf{H}_{nn}$, i.e.,
$\mathbf{H}_{nn}=\mathbf{\widetilde{U}}_{nn}\widetilde{\Sigma}_{nn}
[\mathbf{\overline{b}}_{nn}^1,\cdots,\mathbf{\overline{b}}_{nn}^{M}]^{\dag}.$
The eigenvalues in $\widetilde{\Sigma}_{nn}$ are sorted in 
decreasing order along the diagonal.

\subsubsection{Decorrelator Design for Pco-MIMO}
\textcolor{black}{First, we design the decorrelator $\mathbf{u}_{(k,p)}^i$ at MS $k$ for the $i$-th desired common stream $x_{(k,c)}^i$. To  eliminate the residual interference from the remaining $d_{n,p}-\textcolor{black}{d_{K,M,N}}$ private streams to MS $n\in \mathcal K, n\neq k$ and detect   $x_{(k,c)}^i$,  the
decorrelator $\mathbf{u}_{(k,c)}^i\in\mathbb{C}^{N\times1}$ at MS $k$ is given
by:}
\begin{equation}\label{eq:u_k} \blue
\mathbf{u}_{(k,c)}^i=\mathbf{\widetilde{U}}_{(k,c)}
(\mathbf{\widetilde{U}}_{(k,c)})^{\dag}\mathbf{H}_{kk}\mathbf{b}_{(k,c)}^i/\psi,
\end{equation}
where the columns of
$\mathbf{\widetilde{U}}_{(k,c)}\in\mathbb{C}^{N\times
(M-d_{(k,c)}-\textcolor{black}{\sum_{n\in \mathcal K}}d_{(n,p)}+1)}$ \textcolor{black}{form} the orthonormal basis of
$\text{Null}(\mathbf{\widetilde{H}}_{(k,c)})$  and 
\textcolor{black}{$\mathbf{\widetilde{H}}_{(k,c)}=\left[\mathbf{H}_{k}\mathbf{B}_{(k,c),i}^{(1)},
\mathbf{H}_{kk}\mathbf{B}_{(k,p)}^{(1)},
\mathbf{H}_{k1}\mathbf{B}_{(1,p)}^{(2)},\cdots, \mathbf{H}_{kK}\mathbf{B}_{(K,p)}^{(2)}
\right].$ $\mathbf{B}_{(k,c),i}^{(1)}$ is $\mathbf{B}_{(k,c)}^{(1)}$ without the $i$-th column, i.e., $\mathbf{B}_{(k,c),i}^{(1)}=\left[\mathbf{b}_{(k,c)}^{1},\cdots,
\mathbf{b}_{(k,c)}^{i-1},\mathbf{b}_{(k,c)}^{i+1},\cdots,
\mathbf{b}_{(k,c)}^{d_{(k,c)}}\right]$.} Here, $\psi=||\mathbf{\widetilde{U}}_{(k,c)}
(\mathbf{\widetilde{U}}_{(k,c)})^{\dag}\mathbf{H}_{kk}\mathbf{b}_{(k,c)}^i||$ is to  normalize $\mathbf{u}_{(k,c)}^i$, i.e., $ ||\mathbf{u}_{(k,c)}^i||=1 $. By using the \textcolor{black}{decorrelator} in \eqref{eq:u_k}, the interference is nulled due to the fact that  $ (\mathbf{\widetilde{U}}_{(k,c)})^{\dag}\mathbf{\widetilde{H}}_{(k,c)}=0 $. The equivalent channel for $x_{(k,c)}^i$ is  $ (\mathbf{u}_{(k,c)}^i)^{\dag}\mathbf{H}_{kk}\mathbf{b}_{(k,c)}^i\sqrt{P_{(k,c)}^i}=\frac{1}{\psi}||(\mathbf{\widetilde{U}}_{(k,c)})^{\dag} \mathbf{H}_{kk}\mathbf{b}_{(k,c)}^i||^2\sqrt{P_{(k,c)}^i} .$  The decorrelator
$\mathbf{u}_{(k,p)}^i\in\mathbb{C}^{N\times1}$ at MS $k$  can be designed in a similar way to $\mathbf{u}_{(k,c)}^i$. 
 
 \textcolor{black}{
\begin{Rem}[Flexible Adjustment of Cooperation Level in Pco-MIMO] The Pco-MIMO design is flexible to adjust
the cooperation level  between the coordination and cooperation modes. It also  incorporates the full coordinative MIMO (by choosing $d_{(k,c)}=0$ for all $k\in \mathcal K$), the Uco-MIMO for $K=2$ (by choosing $d_{(k,c)}=0$ and $d_{(n,c)}>0$, where $k,n\in \{1,2\}, n\neq k$) and the full cooperative MIMO (by choosing $d_{(k,p)}=0$ for all $k\in \mathcal K$) as special cases.
%Specifically, by choosing $d_{(k,c)}=0$ for all $k\in \mathcal K$, the ZF precoder and decorrelator design for the private streams of each MS  in the proposed
%Pco-MIMO scheme is the same as that in the coordinative MIMO scheme. By choosing $d_{(k,c)}=0$ and $d_{(n,c)}>0$, where $n\neq k$,  the ZF precoder and decorrelator design for the private streams of each MS  and the ZF precoder and decorrelator design for the common streams of MS $n$ are the same as those in the  Uco-MIMO scheme for $K=2$. By choosing $d_{(k,p)}=0$ for all $k\in \mathcal K$, the ZF precoder and decorrelator design for the common streams of each MS  in the proposed
%Pco-MIMO scheme is the same as that in the full cooperative MIMO scheme.
~\hfill\IEEEQED\label{Rem:BL}
 \end{Rem}}

\subsection{DoF Performance under Perfect CSIT}\label{sec:sub_DOF}
\textcolor{black}{We derive
the system DoF of the Pco-MIMO scheme under the perfect CSIT.}
\textcolor{black}{
\begin{Thm}[DoF Performance of Pco-MIMO]\label{thm:DoF}
Suppose the  backhaul consumption satisfies $R_{(k,c)} = d_{(k,c)} \log_2
(\text{SNR})$ for all $k\in \mathcal K$. The
system DoF under the Pco-MIMO scheme $\text{DoF}(\text{Pco-MIMO})=\sum_{k\in\mathcal K} (d_{(k,p)}+d_{(k,c)})$  satisfies
\begin{equation}
\text{DoF}(\text{Pco-MIMO})\leq 
\text{DoF}_{\max}(\text{Pco-MIMO})\triangleq N+(K-1)d_{K,M,N}+\min_{k\in \mathcal K}\sum_{n\in \mathcal K, n\neq k}d_{(n,c)},\label{eqn:DoF-ineq-thm}
\end{equation}
where the maximum system DoF  $\text{DoF}_{\max}(\text{Pco-MIMO})$ can be achieved when
\begin{align}
\ d_{(k,c)}+\sum_{n\in \mathcal K}d_{(n,p)} =
N+(K-1)d_{K,M,N}, \quad \forall k\in\arg\min_{k\in \mathcal K}\sum_{n\in \mathcal K, n\neq k}d_{(n,c)}.\label{eqn:DoF-eq-thm}
\end{align}
Furthermore, 
$\text{DoF}_{\max}(\text{Pco-MIMO})\leq KN$, where the equality holds when 
\begin{align}
d_{(n,c)}= N- d_{K,M,N}, \quad \forall n \in \mathcal K, n\neq k, k\in\arg\min_{k\in \mathcal K}\sum_{n\in \mathcal K, n\neq k}d_{(n,c)}.
\label{eqn:DoF-eq-max-thm}
\end{align}~\hfill\IEEEQED
\end{Thm}
}
\begin{proof}
Please refer to  Appendix A.
\end{proof}

\textcolor{black}{From Theorem \ref{thm:DoF}, we can see that  the proposed
Pco-MIMO scheme allows a flexible tradeoff between the achievable
DoF and the backhaul consumption.}

%
%\begin{table}
%\begin{center}
%\begin{tabular}{|l|l|l|l|l|}
%\hline
% Scheme & Coordinative MIMO & Uco-MIMO & Pco-MIMO & Cooperative MIMO \\
%  \hline
%DoF ($K=2$)& $M$ & $M$ & $
%M+\min(d_{(1,c)},d_{(2,c)})$ & $2N$\\
%&  & & $\in\{M,\cdots, 2N\}$ & \\
%  \hline
%  DoF ($K>2$)& $N+(K-1)d_{K,M,N}$ & \quad \  not & $
%N+(K-1)d_{K,M,N}+\min\limits_{k\in \mathcal K}\sum\limits_{\substack{n\in \mathcal K\\ n\neq k}}d_{(n,c)}$ & $KN$\\
% &  &    applicable& $\in \{N+(K-1)d_{K,M,N},\cdots, KN\}$ & \\
%\hline
%\end{tabular}
%  \caption{\textcolor{black}{Illustration of the system DoFs of different schemes.}}
%  \label{Tab:DoF}
%\end{center}
%\end{table}

\begin{table}
\begin{center}
\begin{tabular}{|l|l|l|l|l|}
\hline
 Scheme & Coordinative MIMO & Uco-MIMO & Pco-MIMO & Cooperative MIMO \\
  \hline
DoF ($K=2$)& $M$ & $M$ & $
\{M,\cdots, 2N\}$ & $2N$\\
  \hline
  DoF ($K>2$)& $N+(K-1)d_{K,M,N}$ &  not applicable & $
\{N+(K-1)d_{K,M,N},\cdots, KN\}$ & $KN$\\
\hline
\end{tabular}
  \caption{\textcolor{black}{Illustration of the system DoFs of different schemes at  $d_{(k,c)}=d$ for all $k\in\mathcal K$.}}
  \label{Tab:DoF}
\end{center}
\end{table}

\textcolor{black}{
\begin{Rem}[Comparisions of System DoFs]
Table \ref{Tab:DoF} compares the system DoFs of different schemes. (Note that $d_{K,M,N}=M-N$ when $K=2$.)  By choosing $d_{(k,c)}=d\leq  N- d_{K,M,N}$ for all $k\in \mathcal K$, the proposed
Pco-MIMO scheme can
achieve $\text{DoF} (\text{Pco-MIMO})= N+(K-1)d_{K,M,N}+(K-1)d$, i.e., an increase of $(K-1)d$ compared with 
the coordinative MIMO or the uni-directional cooperative MIMO (Uco-MIMO) \cite{MIMO:int:jafar,MIMO:int:huang} (applicable for  $K=2$ only).  This increase is achieved at the cost of the backhaul
consumption of $(K-1)d\log_2 (\text{SNR})$ (for each BS).   By choosing $d_{(k,c)}=d= N- d_{K,M,N}$ for all $k\in \mathcal K$, the proposed
Pco-MIMO scheme can
achieve $\text{DoF} (\text{Pco-MIMO})= KN$, which is the same as the 
full cooperative MIMO, but save backhaul consumption by $(K-1)(N-d)\log_2 (\text{SNR})=(K-1)d_{K,M,N}\log_2 (\text{SNR})\geq0$ (for each BS) compared with the full cooperative MIMO.~\hfill\IEEEQED
\end{Rem}}

\subsection{Mutual Information, System Goodput  \textcolor{black}{under Imperfect CSIT}}
\label{sec:imp_CSI}

When the CSIT is imperfect,
there will be \textcolor{black}{residual} interference at \textcolor{black}{each MS} due to imperfect
\textcolor{black}{ZF} precoding. Given the decorrelator
$\mathbf{u}_{(k,c)}^i$ of the $i$-th common
stream, the \textcolor{black}{recovered} signal of $x_{(k,c)}^i$ in \eqref{eqn:recover-r} at MS $k$ is given by:
\begin{equation}\label{eq:y_k_u}\begin{array}{lll}
\textcolor{black}{r_{(k,c)}^i}=(\mathbf{u}_{(k,c)}^i)^{\dag}\mathbf{y}_k=
\underbrace{(\mathbf{u}_{(k,c)}^i)^{\dag}\mathbf{H}_{kk}\mathbf{b}_{(k,c)}^i\sqrt{P_{(k,c)}^i}x_{(k,c)}^i}_{\text{desired
signal for MS $k$} }+\underbrace{ (\mathbf{u}_{(k,c)}^i)^{\dag}\mathbf{I}_{k}
}_{\text{residual interference at MS $k$}}+
(\mathbf{u}_{(k,c)}^i)^{\dag}\mathbf{z}_k,
\end{array}\nonumber
\end{equation}
where
$\mathbf{I}_{k}=\textcolor{black}{\sum_{n\in \mathcal K, n\neq k}}\left(\mathbf{H}_{kn}\mathbf{B}_{(n,p)}^{(1)}\Sigma_{(n,p)}^{(1)}\mathbf{x}_{(n,p)}^{(1)}+
\mathbf{H}_{k}\mathbf{B}_{(n,c)}\Sigma_{(n,c)}\mathbf{x}_{(n,c)}\right)$. Assuming {\red Gaussian}
inputs for the system {\red  and treating interference as noise}, the  mutual information
(bit/s/Hz) of the $i$-th common stream at MS $k$  is given by:
\begin{equation}\label{eq:capacity_c}
C_{(k,c)}^i = \log_{2}\left( 1+\sigma_{(k,c)}^i
P_{(k,c)}^i/(1+I_{(k,c)}^i) \right), \forall i=1,\cdots,d_{(k,c)},
\end{equation}
where
$\sigma_{(1,c)}^i=|(\mathbf{u}_{(k,c)}^i)^\dag\mathbf{H}_{kk}\mathbf{b}_{(k,c)}^i|^2$
and $$I_{(k,c)}^i=\textcolor{black}{\sum_{n\in \mathcal K, n\neq k}}\left(\sum_{i=1}^{\textcolor{black}{d_{K,M,N}}}P_{(n,p)}^i
|(\mathbf{u}_{(k,c)}^i)^{\dag}\mathbf{H}_{kn}\mathbf{b}_{(n,p)}^i|^2+
\sum_{i=1}^{d_{(n,c)}}P_{(n,c)}^i\cdot
|(\mathbf{u}_{(k,c)}^i)^{\dag}\mathbf{H}_{k}\mathbf{b}_{(n,c)}^i|^2\right).$$
Similarly, the  mutual information (bit/s/Hz) of the
$i$-th private stream at MS $k$ is given by:
\begin{equation}\label{eq:capacity_p}
C_{(k,p)}^i = \log_{2}\left( 1+\sigma_{(k,p)}^i
P_{(k,p)}^i/(1+I_{(k,p)}^i) \right), \forall i=1,\cdots,d_{(k,p)},
\end{equation}
where
$\sigma_{(1,p)}^i=|(\mathbf{u}_{(k,p)}^i)^\dag\mathbf{H}_{kk}\mathbf{b}_{(k,p)}^i|^2$
and $I_{(k,p)}^i$ is calculated in a similar way to $I_{(k,c)}^i$.

Due to the imperfect CSIT, the  mutual information
$C_{(k,c)}^i$ and $C_{(k,p)}^i$ {\red at a frame} are not completely known to the BSs. Thus, there will be packet errors when the \textcolor{black}{transmit} data
rate exceeds the mutual information. Let $R_{(k,c)}$ and $R_{(k,p)}$
be the scheduled data rate of the common streams and the private streams
of BS $k$, respectively. (\textcolor{black}{Note that $R_{(k,c)}$
 also indicates the backhaul
consumption for sharing  common streams from BS $k$.}) The goodput   $U_k$ at MS $k$ (number of bits
successfully received) \textcolor{black}{in one frame}  is given by:
\begin{equation}\label{eq:Ubit_k} \blue
U_k=\tau R_{(k,c)}\mathbf{1}\big(R_{(k,c)}\leq C_{(k,c)} \big)+\tau
R_{(k,p)}\mathbf{1}\big(R_{(k,p)}\leq C_{k,p}\big),
\end{equation}
where $C_{(k,c)}=\sum_{i=1}^{d_{(k,c)}}C_{(k,c)}^i$,
$C_{k,p}=\sum_{i=1}^{d_{(k,p)}}C_{(k,p)}^i$  and $\mathbf 1 (\cdot)$
 denotes the indicator function.

The total transmit power of BS $k$ to support the $d_{(k,p)}$ private streams, the $d_{(k,c)}$ common streams to MS $k$ and  the $d_{(n,c)}$   common streams to MS $n\in \mathcal K, n\neq k$ is given by:
\begin{equation}\label{eq:p_tx}\begin{array}{l}
P_k^{tx}= \sum\nolimits_{i=1}^{d_{(k,p)}}P_{(k,p)}^i+\sum_{n\in \mathcal K}P_{(n,c),k},
\end{array}
\end{equation}
\textcolor{black}{where $ P_{(n,c),k}=\sum\nolimits_{i=1}^{d_{(n,c)}}P_{(n,c)}^i\alpha_{(n,k)}^i$  denotes
the transmit power at BS $k$ for the common streams to MS $n$} and $\alpha_{(n,k)}^i= \sum_{m=1}^M\big|[\textcolor{black}{\mathbf{B}_{(n,c)}}]_{((k-1)M+m,i)}\big|^2$ denotes the portion of power $P_{(n,c)}^i$  for common stream  $x_{(n,c)}^i$ contributed by BS $k$. Note that each common stream is precoded at the \textcolor{black}{$K$} BSs, and hence, we have  \textcolor{black}{$\sum_{k\in \mathcal K}\alpha_{(n,k)}^i=1$ for all $i=1,\cdots, d_{(n,c)}$, $n\in \mathcal K$}.

\section{Delay-Optimal Cross Layer Resource Optimization}\label{sec:delay-optimal-form}

  In this
section, we  formally define the control policy and formulate the delay-optimal  control problem under average power and  backhaul constraints.

%In this section, we  \textcolor{black}{formulate the delay-optimal control problem for} the
%proposed Pco-MIMO scheme  \textcolor{black}{to support}  delay-sensitive applications.
%\subsection{Motivation}
%%The DoF \textcolor{black}{increases} via data sharing \textcolor{black}{of} the common streams comes
%%with the price of an increase in backhaul consumption. When we do
%%not have backhaul constraints, it is obvious that the optimal choice
%%is to share all the data streams between the two BSs so that there
%%is no more private streams. This is because the common streams can
%%benefit from the joint processing of two BSs while the private
%%stream can only benefit from local processing of one BS. However,
%%When there \textcolor{black}{are backhaul constraints},
%%the {\em reward} of
%%engaging in BS cooperation for a common data stream is higher
%%goodput but the {\em price} is the corresponding backhaul
%%consumption.
%Under the backhaul constraints, the
%decision on the data sharing for delay-sensitive applications should
%be adaptive to the \textcolor{black}{CSIT and} the
%QSI. The CSIT \textcolor{black}{and QSI} adaptation allow
%the design to capture {\em good channel opportunity}  and  the {\em urgency} of the traffic demand for the
%cooperation between the two BSs at the expense of the backhaul consumption.  In the
%following, we  formally define the control policy, and \textcolor{black}{formulate} the delay-optimal \textcolor{black}{control problem}.

\subsection{Control Policy and Resource Constraints}
Denote $\boldsymbol{\chi}=\{\mathbf{H},
\mathbf{Q}\}$ as the global system state and
$\hat{\boldsymbol{\chi}}=\{\mathbf{\hat{H}}, \mathbf{Q}\}$ as the observed  
global system state. \textcolor{black}{The complete system state is $\{\boldsymbol{\chi},\hat{\boldsymbol{\chi}}\}$.} Based on the Pco-MIMO scheme,  at the beginning of each frame,   determine the transmit
power  and  rate allocation  \textcolor{black}{based on the  
global observed system state $\hat{\boldsymbol{\chi}}$}
according to the following  stationary  control policy.
\begin{Def}[Stationary Power and Rate Control Policy]\label{def:policy} A stationary power and rate control
policy \textcolor{black}{$\Omega=\{\Omega_P,\Omega_R\}$} is a mapping from the
observed state $\hat{\boldsymbol{\chi}}$ to the power and rate
allocation actions  \textcolor{black}{$\Omega(\hat{\boldsymbol{\chi}})=\{\Omega_P(\hat{\boldsymbol{\chi}}),\Omega_R(\hat{\boldsymbol{\chi}})\}$}, where
$\Omega_P(\hat{\boldsymbol{\chi}})=\mathbf{P}=\{\mathbf{P}_{(k,p)},\mathbf{P}_{(k,c)}:
k\in \mathcal K\}$ and
$\Omega_R(\hat{\boldsymbol{\chi}})=\mathbf{R}=\textcolor{black}{\{\mathbf R_k:k\in \mathcal K\}}$. $\mathbf{P}_{(k,p)}=\{P_{(k,p)}^i\in\mathbb{R}^+: i=1,\cdots, d_{k,p}\}$, $\mathbf{P}_{(k,c)}=\{P_{(k,c)}^i\in\mathbb{R}^+: i=1,\cdots, d_{k,c}\}$, \textcolor{black}{and $\mathbf R_k=\{R_{(k,p)},R_{(k,c)}\in\mathbb{R}^+\}$.} Assume  $\Omega$ is unichain\footnote{{\red  Unichain policy is a special type of stationary policy, for which the corresponding Markov chain $\left\{\boldsymbol{\chi}(t),\hat{\boldsymbol{\chi}}(t)\right\}$ has a single recurrent class (and possibly some transient states) \textcolor{black}{\cite[Chap. 4]{Bertsekas:2007}}.}}.~\hfill\IEEEQED
\end{Def}

The power allocation policy $\Omega_P$ satisfies the
per-BS average power consumption constraint:
\begin{eqnarray}\label{eq:pwr_con}
\overline{P}_k(\Omega)=\lim\sup_{T\rightarrow\infty}\frac{1}{T}\sum_{t=1}^T\mathbb{E}^{\Omega}[P_k(t)]\leq
P_k^0, \quad \forall k\in \mathcal K,
\end{eqnarray}
where {\red $\mathbb{E}^{\Omega}$ indicates that the expectation is taken w.r.t. the measure induced by the policy $\Omega$,} $P_k(t)$ is the total power consumption of BS $k$ at frame $t$ given
in (\ref{eq:p_k}), and  \textcolor{black}{$P_k^0$ denotes the maximum average power consumption}. On the other hand, the  rate allocation policy $\Omega_R$
satisfies the average backhaul  consumption constraint:
\begin{eqnarray}\label{eq:bkh_con}
\overline{R}_{(k,c)}(\Omega)=\lim\sup_{T\rightarrow\infty}\frac{1}{T}\sum\nolimits_{t=1}^T\mathbb{E}^{\Omega}[R_{(k,c)}(t)]\leq
R_{(k,c)}^0,\quad \forall k\in \mathcal K,
\end{eqnarray}
where $R_{(k,c)}(t)$ is the scheduled data rate for the common
streams $\mathbf{x}_{(k,c)}$ at frame $t$ and  $R_{(k,c)}^0$ denotes the maximum average backhaul consumption.\footnote{Note that the backhaul constraints  \textcolor{black}{account}
for the backhaul consumption due to the data sharing only. The backhaul
consumption for  the CSIT sharing  is negligible
compared with that for the data sharing. This is because the CSIT sharing is done once per frame
while the data sharing is done once per symbol.}

\subsection{Problem Formulation}
For a given stationary control policy $\Omega$, the induced random
process \textcolor{black}{$\left\{\boldsymbol{\chi}(t),\hat{\boldsymbol{\chi}}(t)\right\}$} is a controlled Markov chain with the
transition probability given by\footnote{
Note that the equality is due to the independence between  $\mathbf H(t+1),\hat{\mathbf H}(t+1)$ and  $\mathbf Q(t+1)$,  the i.i.d. assumption of the CSI model, the assumption of the imperfect CSIT model and  the independence between  $\mathbf Q(t+1)$ and $\hat{\mathbf H}(t)$  conditioned on $\boldsymbol{\chi}(t)$ and  $\Omega(\hat{\boldsymbol{\chi}}(t))$.
%Note that the first equality is due to the independence between  $\mathbf H(t+1),\hat{\mathbf H}(t+1)$ and  $\mathbf Q(t+1)$. The second equality is due to the i.i.d. assumption of the CSI model, the assumption of the imperfect CSIT model and  the independence between  $\mathbf Q(t+1)$ and $\hat{\mathbf H}(t)$  conditioned on $\boldsymbol{\chi}(t)$ and  $\Omega(\hat{\boldsymbol{\chi}}(t))$.
}:
\begin{align}
&\Pr\{\boldsymbol{\chi}(t+1),\hat{\boldsymbol{\chi}}(t+1)|\boldsymbol{\chi}(t),\hat{\boldsymbol{\chi}}(t),\Omega(\hat{\boldsymbol{\chi}}(t))\}\nonumber\\
%=&\Pr\{\mathbf H(t+1),\hat{\mathbf H}(t+1)|\boldsymbol{\chi}(t),\hat{\boldsymbol{\chi}}(t),\Omega(\hat{\boldsymbol{\chi}}(t))\}\Pr\{\mathbf Q(t+1)|\boldsymbol{\chi}(t),\hat{\boldsymbol{\chi}}(t),\Omega(\hat{\boldsymbol{\chi}}(t))\}\nonumber\\
=&\Pr\{\mathbf{H}(t+1)\}\Pr\{\hat{\mathbf{H}}(t+1)|\mathbf{H}(t+1)\}\Pr\{\mathbf{Q}(t+1)|\boldsymbol{\chi}(t),\Omega(\hat{\boldsymbol{\chi}}(t))\},\label{eq:sys_tran}
\end{align}
where the queue transition probability is given by
\begin{equation}\label{eq:Q_tran}
\Pr\{\mathbf{Q}(t+1)|\boldsymbol{\chi}(t),\Omega(\hat{\boldsymbol{\chi}}(t))\}=
\prod_{k\in \mathcal K}\Pr\left\{A_{k}(t)\in \mathbb R^+:A_{k}(t)\text{ satisfies } \eqref{eq:Q_org}\right\}
%\left\{
%\begin{array}{lll}
%\prod_{k}\Pr\{A_{k}(t)\}& \text{if \textcolor{black}{\eqref{eq:Q_org} holds for all $k\in \mathcal K$} } \\
%0 & \text{ otherwise}
%\end{array}\right..
\end{equation}
%where $\hat{Q}_k=\Big[\big[Q_k(t)-U_k(t)\big]^+ +
%A_k(t)\Big]_{\bigwedge N_Q}$ and $U_k(t)$ is given in
%\eqref{eq:Ubit_k}.
Note that, the stochastic dynamics of the \textcolor{black}{$K$} queues are coupled
together via  $\Omega$.

Given a stationary control policy $\Omega$, the average  delay cost\footnote{\textcolor{black}{The average delay cost defined here is a general queue size dependent metric, which includes the average delay as a special case.}}  of  \textcolor{black}{MS} $k$ is given by:
\begin{equation}
\label{eq:T_single} \overline{T}_{k}(\Omega)=\lim\sup_{T\rightarrow
\infty}\frac{1}{T}\sum\nolimits_{t=1}^T\mathbb{E}^{\Omega}\left[f(Q_{k}(t))\right],\quad \forall k\in \mathcal K,
\end{equation}
where $f(Q_{k})$ is a monotonic increasing  \textcolor{black}{cost} function of
$Q_{k}$. For example, when $f(Q_{k})=Q_{k}/\lambda_{k}$,  by
Little's Law \cite{Ross:2003}, $\overline{T}_{k}(\Omega)$ is the
average delay  of user $k$. When $f(Q_{k})=\mathbf{1}(Q_k\geq
Q_k^0)$, $\overline{T}_{k}(\Omega)$ is the probability that \textcolor{black}{$Q_k$}  exceeds $Q_k^0$ for some reference
$Q_k^0\in\{0,\cdots,N_Q\}$.

For some positive constants $\boldsymbol{\beta}=\{\beta_{k}:
k\in \mathcal K\}$,  define the average weighted sum delay cost under  a stationary control policy $\Omega$ as: 
$$\overline{T}_{\beta}^{\Omega}\triangleq \sum_{k\in \mathcal K}\beta_{k}\overline{T}_{k}(\Omega)=\lim\sup_{T\rightarrow
\infty}\frac{1}{T}\sum\nolimits_{t=1}^T\mathbb{E}^{\Omega}\left[\sum_{k\in \mathcal K}\beta_{k}f(Q_k\textcolor{black}{(t)})\right]\textcolor{black}{.}$$

The delay-optimal control problem is formulated as follows\footnote{The positive constants
$\boldsymbol{\beta}$ indicate the relative importance of the users. \textcolor{black}{For given $\boldsymbol{\beta}$}, the solution to Problem \ref{prob:delay} corresponds to a Pareto optimal point of the
multi-objective optimization problem given by $\min_{\Omega}
\overline{T}_{k}(\Omega)\ s.t. \ \eqref{eq:pwr_con} \text{
 and } \eqref{eq:bkh_con}$  for all $k\in \mathcal K$.}: 
 \begin{Prob}[Delay-optimal Control Problem for Pco-MIMO]\label{prob:delay}
\begin{equation}
\begin{array}{l}
\min_{\Omega}\quad \overline{T}_{\beta}^{\Omega}\\
s.t.\quad \text{the average power and backhaul constraints in
\eqref{eq:pwr_con} and \eqref{eq:bkh_con} \textcolor{black}{for all $k\in \mathcal K$}}\textcolor{black}{.}
\end{array}\nonumber
\end{equation}
\end{Prob}
 \textcolor{black}{Note that under the time average expected constraints\footnote{\textcolor{black}{The time averaged objective and constraints are  commonly used in the literature. For example, the egordic capacity maximization (the average delay minimization) under the average power constraint \cite{Tse04fundamentals} (\cite{Berry:2002}).}} in \eqref{eq:pwr_con} and \eqref{eq:bkh_con}, the probability, that the instantaneous power and backhaul consumption goes to infinity, goes to zero.  Furthermore,  additional peak power or backhaul consumption constraints can be accommodated in Problem \ref{prob:delay}.}

%\begin{Rem}[Interpretation of $\boldsymbol{\beta}$]
%The positive constants
%$\boldsymbol{\beta}$ indicate the relative importance of the users. \textcolor{black}{For given $\boldsymbol{\beta}$}, the solution to Problem \ref{prob:delay} corresponds to a Pareto optimal point of the
%multi-objective optimization problem given by $\min_{\Omega}
%\overline{T}_{k}(\Omega)\ s.t. \ \eqref{eq:pwr_con} \text{
% and } \eqref{eq:bkh_con}$  for all $k\in \mathcal K$.~\hfill\IEEEQED
% \end{Rem}

\textcolor{black}{
\begin{Rem} [Interpretation of Problem \ref{prob:delay}] Problem \ref{prob:delay} is an 
infinite horizon constrained average cost per stage problem \cite[Chap.4]{Bertsekas:2007} or constrained Markov decision process (MDP)\cite{Borkaractorcritic:2005}. Specifically, 
since the control policy is defined on the observed system state $\hat{\boldsymbol{\chi}}$ instead of the complete system state  $\{\boldsymbol{\chi},\hat{\boldsymbol{\chi}}\}$, Problem \ref{prob:delay} belongs to 
constrained partially observed MDP (CPOMDP),
which is well-known to be a very difficult problem
\cite{Meuleau:1999}.~\hfill\IEEEQED 
\end{Rem}}

\section{General Solution to the Delay Optimal Problem}\label{sec:opt_solution}
 In this section, by exploiting the special structure in our
problem, we derive an {\em equivalent Bellman equation} to
simplify the  CPOMDP problem.

\textcolor{black}{We} consider the dual problem of the CPOMDP  in Problem
\ref{prob:delay}. For any nonnegative Lagrange multipliers  (LMs) \textcolor{black}{
$\boldsymbol{\gamma}=\{\gamma_{(k,P)},\gamma_{(k,C)}\in \mathbb R^+:k\in \mathcal K\}$}, define the
Lagrangian as $L_{\beta}(\Omega,\boldsymbol{\gamma})=\lim\sup_{T\rightarrow\infty}
\frac{1}{T}\sum\nolimits_{t=1}^T \mathbb{E}^{\Omega}\left[
g\left(\boldsymbol{\gamma},\boldsymbol{\chi}(t),\Omega(\hat{\boldsymbol{\chi}}(t))\right)\right]$, 
%\begin{equation}
%L_{\beta}(\Omega,\boldsymbol{\gamma})=\lim\sup_{T\rightarrow\infty}
%\frac{1}{T}\sum\nolimits_{t=1}^T \mathbb{E}^{\Omega}\left[
%g\left(\boldsymbol{\gamma},\boldsymbol{\chi}(t),\Omega(\hat{\boldsymbol{\chi}}(t))\right)\right],\nonumber
%\end{equation}
where
$g(\boldsymbol{\gamma},\boldsymbol{\chi},\Omega(\hat{\boldsymbol{\chi}}))
= \sum_{k\in \mathcal K} \big(\beta_{k} f(Q_{k})+
\gamma_{(k,P)}$ $(P_{k}-\textcolor{black}{P_k^0})+\gamma_{(k,C)}(R_{(k,c)}-\textcolor{black}{R_{(k,c)}^0}) \big).$ The associated dual problem of Problem \ref{prob:delay} is given by
 \begin{align}
 \max_{\boldsymbol{\gamma}\succeq0}G(\boldsymbol{\gamma}),\label{eqn:L-dual-prob}
 \end{align}
where the Lagrange dual function (unconstrained POMDP) is given by:
\begin{eqnarray}\label{eq:dual}
G(\boldsymbol{\gamma})=\min_{\Omega}L_{\beta}(\Omega,\boldsymbol{\gamma}). 
\end{eqnarray}
We discuss the solution to the dual problem in \eqref{eqn:L-dual-prob} and  the duality gap below.

%First, we define partitioned actions below.
%\begin{Def}[Partitioned Actions]\label{def:con_action}
%Given a control policy $\Omega$, we define $\Omega(\mathbf{Q})=\{
%(\mathbf{P},\mathbf{R}) = \Omega(\mathbf{Q},\hat{\mathbf{H}}): \forall
%\hat{\mathbf{H}}\}$ as the collection of actions for all possible
%CSIT $\hat{\mathbf{H}}$ conditioned on  given QSI $\mathbf{Q}$. The
%complete control policy $\Omega$ is therefore equal to the union of
%all partitioned actions, i.e., $\Omega =
%\bigcup_{\mathbf{Q}}\Omega(\mathbf{Q})$.
%~\hfill\IEEEQED
%\end{Def}

While POMDP is a difficult problem in general, we  utilize the i.i.d. assumption of the
CSI  to substantially simplify the unconstrained POMDP in  \eqref{eq:dual}.  The optimal control policy $\Omega^*$,
can be obtained by solving an equivalent \textcolor{black}{optimality} equation, which is  
summarized below.

%Furthermore,  post-decision state is defined to be the virtual state
%immediately after making an action but before the new bits arrive
%\footnote{The post-decision state framework is used in
%\cite{Thesis:Salodkar} and the references therein.}. For example,
%$\mathbf{Q}$ is the state at the beginning of some time frame (also
%called the {\em pre-decision state}), and making an action
%$\Omega(\hat{\boldsymbol{\chi}})=\{\mathbf{P},\mathbf{R}\}$, the
%post-decision state immediately after the action is
%$\widetilde{\mathbf{Q}}$, where the transition to
%$\widetilde{\mathbf{Q}}$ is given by
%$\widetilde{\mathbf{Q}}=\big(\mathbf{Q}-\mathbf{U}\big)^+$. If new
%arrivals $\mathbf{A}$ occur in the post-decision state, then the
%system reaches the next actual state, i.e., pre-decision state,
%$\mathbf{Q}^{\prime}=\big[\widetilde{\mathbf{Q}}+\mathbf{A}\big]_{\bigwedge
%N_Q}$.

\begin{Thm}[Equivalent Bellman Equation]\label{thm:MDP_post}
$\quad$

 (a)  For any given LMs $\boldsymbol \gamma$, the optimal
    control policy $\Omega^*=(\Omega_P^*, \Omega_R^*)$ for \textcolor{black}{the unconstrained POMDP} in \eqref{eq:dual} can be
    obtained by solving the following {\em
        equivalent Bellman equation} \textcolor{black}{w.r.t. $\theta$ and $\{V(\widetilde{\mathbf{Q}})\}$}:
\begin{align}
\label{eq:bellman_post}
V(\widetilde{\mathbf{Q}})+\theta=&\mathbb E\left[ 
\min_{\mathbf P, \mathbf R}\mathbb E\left[g(\boldsymbol{\gamma},\boldsymbol{\chi},\mathbf P, \mathbf R)+\sum\nolimits_{\widetilde{\mathbf{Q}}^{\prime}}\Pr\{\widetilde{\mathbf{Q}}^{\prime}|\boldsymbol{\chi},\mathbf P, \mathbf R\}V(\widetilde{\mathbf{Q}}^{\prime})\bigg| \hat{\boldsymbol{\chi}}
\right]\Bigg|\widetilde{\mathbf{Q}}\right], \  \forall \widetilde{\mathbf{Q}} \in \boldsymbol{\mathcal{Q}},
\end{align}
where $\theta =G(\boldsymbol{\gamma})$ is the optimal average cost per stage and \textcolor{black}{$V(\cdot)$ is the post-decision state potential function}. 
%$\widetilde{g}(\boldsymbol{\gamma},\mathbf{Q},\Omega(\mathbf{Q}))=
%\mathbb{E}[g(\boldsymbol{\gamma},\boldsymbol{\chi},\Omega(\hat{\boldsymbol{\chi}}))|\mathbf{Q}]$
%is the per-stage \textcolor{black}{cost} function, and
%$\Pr\{\widetilde{\mathbf{Q}}^{\prime}|\mathbf{Q},\Omega(\mathbf{Q})\}=\textcolor{black}{\mathbb{E}
%\left[\Pr\{\widetilde{\mathbf{Q}}^{\prime}|\boldsymbol{\chi},\Omega(\hat{\boldsymbol{\chi}})\}|\mathbf{Q}\right]}$ is the transition kernel. 
$\widetilde {\mathbf Q}$ is the post-decision state, $\mathbf Q=[\widetilde{\mathbf Q}+\mathbf A]_{\bigwedge N_Q}$ is the pre-decision state, and
$\widetilde{\mathbf{Q}}^{\prime}=(\mathbf{Q}-\mathbf{U})^+$ is the next post-decision state
transited from $\mathbf{Q}$\textcolor{black}{\cite[Chap. 3]{Thesis:Salodkar}}, \footnote{The post-decision  \textcolor{black}{queue} state $\widetilde{\mathbf{Q}}$ is the  queue state immediately after making an action but
before new bits arrive\textcolor{black}{\cite[Chap. 3]{Thesis:Salodkar}}. For example, suppose $\mathbf{Q}$ is the queue state at
the beginning of the current frame (also called the {\em
pre-decision state}). After making an action
$\Omega(\hat{\boldsymbol{\chi}})=\{\mathbf{P},\mathbf{R}\}$  leading to a
goodput of $\mathbf{U}$, the post-decision state immediately after
the action is
$\widetilde{\mathbf{Q}}=\big(\mathbf{Q}-\mathbf{U}\big)^+$. The pre-decision queue state at the beginning of the next frame
is given by
$\mathbf{Q}^{\prime}=\big[\widetilde{\mathbf{Q}}+\mathbf{A}\big]_{\bigwedge
N_Q}$.} \textcolor{black}{where $\mathbf U=\{U_k:k\in \mathcal K\}$.} 

(b) If
$\Omega^*(\hat{\boldsymbol{\chi}})=\arg\min_{\mathbf P, \mathbf R}\mathbb E\left[g(\boldsymbol{\gamma},\boldsymbol{\chi},\mathbf P, \mathbf R)+\sum\nolimits_{\widetilde{\mathbf{Q}}^{\prime}}\Pr\{\widetilde{\mathbf{Q}}^{\prime}|\boldsymbol{\chi},\mathbf P, \mathbf R\}V(\widetilde{\mathbf{Q}}^{\prime})\bigg| \hat{\boldsymbol{\chi}}
\right]$ is unique for all $\hat{\boldsymbol{\chi}}$, then the deterministic policy $\Omega^*$ is the optimal policy for  the unconstrained POMDP in \eqref{eq:dual}.~\hfill\IEEEQED
\end{Thm}
\begin{proof}
Please refer to Appendix B.
\end{proof}

%\begin{Rem}[Interpretation of Theorem \ref{thm:MDP_post}]
%The equivalent Bellman equation in \eqref{eq:bellman_post} is
%defined on   $\mathbf{\widetilde{Q}}$
% only. Nevertheless, the optimal control policy
%$\Omega^*=\{\Omega_P^*,\Omega_R^*\}$ obtained by solving
%\eqref{eq:bellman_post} is still adaptive to the global observed state
%$\hat{\boldsymbol{\chi}}=\{\mathbf{Q},\hat{\mathbf{H}}\}$.~\hfill\IEEEQED
%\end{Rem}

Note that the optimization problem in Problem \ref{prob:delay} is not
convex w.r.t. the control policy $\Omega$. The following lemma
\textcolor{black}{establishes} the zero duality gap between the primal and dual
problems.
\begin{Lem}[Zero Duality Gap]\label{lem:dual_gap} {\red If the condition of Theorem \ref{thm:MDP_post} (b) holds,}  the duality gap between the primal
problem in Problem \ref{prob:delay} and the dual problem {\red in \eqref{eqn:L-dual-prob}} is zero,
i.e.,
\begin{equation}\label{eq:zero_duality}
\min_{\Omega}\max_{\boldsymbol{\gamma}\succeq0}
L_{\beta}(\Omega,\boldsymbol{\gamma})=\max_{\boldsymbol{\gamma}\succeq0}\min_{\Omega}
L_{\beta}(\Omega,\boldsymbol{\gamma}).
\end{equation}
\end{Lem}
\begin{proof}
Please refer to Appendix C.
\end{proof}

Therefore, by solving the dual problem in \eqref{eqn:L-dual-prob}, we can obtain the primal
optimal $\Omega^*$. In other words, the derived \textcolor{black}{policy} of the
equivalent Bellman equation in \eqref{eq:bellman_post} for  dual
optimal LMs $\boldsymbol{\gamma}^*$ {\red solves} the CPOMDP (primal problem) in Problem
\ref{prob:delay}.

\textcolor{black}{
\begin{Rem} [Discussions on Optimal Solution] The
brute-force solution using Theorem \ref{thm:MDP_post} and Lemma \ref{lem:dual_gap} requires solving a  large system of
nonlinear fixed point equations  in
\eqref{eq:bellman_post}. The obtained optimal solution has exponential complexity w.r.t. the number of MSs and requires centralized implementation and knowledge of  system statistics. 
In the following section, we study a low-complexity distributed solution based on the optimal solution. 
\end{Rem}}

\section{Low Complexity Distributed Solution}\label{sec:low-complexity-solution}

In this section, we propose a low-complexity distributed solution using
a distributed online learning algorithm to estimate the per-flow
potential functions and LMs and
a distributed online stochastic  partial gradient algorithm to
obtain the  power and rate control policy.   

%The proposed low complexity distributed solution is
%based on local observations of the system states at each BS and
%does not require explicit knowledge of the system statistics.

 \subsection{Linear Approximation of System Potential Functions}
To reduce \textcolor{black}{computational complexity and} facilitate
distributed implementation, we first approximate
the system post-decision state potential functions \textcolor{black}{$\{V(\widetilde{\mathbf{Q}})\}$} defined  in \eqref{eq:bellman_post}
by the sum of  the per-flow post-decision state potential  functions \textcolor{black}{$\{V_{k}(\widetilde{Q}_{k})\}$ for all $k\in \mathcal K$} below:
\begin{equation}
\label{eq:linear_value}
V(\widetilde{\mathbf{Q}})\approx\sum_{k\in \mathcal K}V_{k}(\widetilde{Q}_{k}), \quad \forall \widetilde{\mathbf{Q}}\in \boldsymbol{\mathcal Q},
\end{equation}
where $\{V_{k}(\widetilde{Q}_{k})\}$ is defined as the {\em
fixed point} of the following per-flow fixed point equation:
\begin{align}
&V_{k}(\widetilde{Q}_{k})+V_{k}(\widetilde{Q}_{k}^0)\label{eq:bellman_perflow}\\
=&
\mathbb E\left[\min_{\mathbf P_k, \mathbf R_k}\left[
g_k(\boldsymbol{\gamma}_{k},\hat{\boldsymbol{\chi}}_k,
\mathbf P_k, \mathbf R_k)+
\sum_{\widetilde{Q}_{k}^{\prime}}\Pr\{\widetilde{Q}_{k}^{\prime}|\hat{\boldsymbol{\chi}}_k, \mathbf P_k, \mathbf R_k\}V_{k}(\widetilde{Q}_{k}^{\prime})\right]\Bigg| \widetilde Q_k\right].\nonumber
\end{align}
\textcolor{black}{$g_k(\boldsymbol{\gamma}_{k},\hat{\boldsymbol{\chi}}_k,
\mathbf P_k, \mathbf R_k)=\beta_k f(Q_{k})+\gamma_{(k,P)} \left(\sum\nolimits_{i=1}^{d_{(k,p)}}P_{(k,p)}^i+P_{\text{cct}}\mathbf{1}(\sum\nolimits_{i=1}^{d_{(k,p)}}P_{(k,p)}^i>0)-P_k^0\right)+ \sum_{n\in \mathcal K, n\neq k}\gamma_{(n,P)}P_{(k,c),n}+\gamma_{(k,C)}(R_{(k,c)}-R_{(k,c)}^0) $, $\boldsymbol{\gamma}_{k}=\{\gamma_{(k,C)},\gamma_{(n,P)}:n\in \mathcal K\}$} and 
%$\widetilde{g}_k(\boldsymbol{\gamma}_{k},Q_{k},{\red \widetilde \Omega_{k}}(Q_{k}))=
%\mathbb{E}[g_k(\boldsymbol{\gamma}_{k},\hat{\boldsymbol{\chi}}_k,
%\widetilde \Omega_{k}(\hat{\boldsymbol{\chi}}_k))|Q_k],$
%and $\Pr\{\widetilde{Q}_{k}^{\prime}|Q_{k}, \widetilde \Omega_{k}(Q_{k})\}=\mathbb{E}
%[\Pr\{\widetilde{Q}_{k}^{\prime}|\hat{\boldsymbol{\chi}}_k, \widetilde \Omega_{k}(\hat{\boldsymbol{\chi}}_k)\}|Q_k].$ 
%Similar to Definition \ref{def:con_action}, $\widetilde  \Omega_{k}(Q_{k})=\{\widetilde  \Omega_{k}(Q_{k},\hat{\mathbf H}):\forall \hat{\mathbf H} \}$.
%\begin{eqnarray}\label{eq:g_k}
%&&\widetilde{g}_k(\boldsymbol{\gamma}_{k},Q_{k},{\red \widetilde \Omega_{k}}(Q_{k}))=
%\mathbb{E}[g_k(\boldsymbol{\gamma}_{k},Q_k,\hat{\mathbf{H}},
%{\red \widetilde \Omega_{k}}(Q_{k},\hat{\mathbf{H}}))|Q_k],\\
%&&\Pr\{\widetilde{Q}_{k}^{\prime}|Q_{k},{\red \widetilde \Omega_{k}}(Q_{k})\}=\mathbb{E}
%[\Pr\{\widetilde{Q}_{k}^{\prime}|Q_{k},\hat{\mathbf{H}}, \widetilde \Omega_{k}(Q_{k},\hat{\mathbf{H}})\}|Q_k].
%\end{eqnarray}
$\hat{\boldsymbol{\chi}}_k=\{Q_k,\hat{\mathbf{H}}\}$.  $\widetilde Q_k$ is the post-decision state, $ Q_k=[\widetilde{ Q}_k+ A_k]_{\bigwedge N_Q}$  is the pre-decision state, and $\widetilde{Q}^{\prime}_k=(Q_k-U_k)^+$ is
the next post-decision state
transited from $Q_k$. 
$\widetilde{Q}_{k}^0\in\{0,\cdots,N_Q\}$ is a
reference state.  Let $\widetilde  \Omega_{k}^*$ denote the policy satisfying  $$\widetilde  \Omega_{k}^*(\hat{\boldsymbol{\chi}}_k)=\arg\min\limits_{ \mathbf P_k, \mathbf R_k} \left[
g_k(\boldsymbol{\gamma}_{k},\hat{\boldsymbol{\chi}}_k,
\mathbf P_k, \mathbf R_k)+
\sum_{\widetilde{Q}_{k}^{\prime}}\Pr\{\widetilde{Q}_{k}^{\prime}|\hat{\boldsymbol{\chi}}_k, \mathbf P_k, \mathbf R_k\}V_{k}(\widetilde{Q}_{k}^{\prime})
\right], \quad \forall \hat{\boldsymbol{\chi}}_k.$$ 

The  linear approximation in \eqref{eq:linear_value} is accurate under certain conditions.
\begin{Lem}[Optimality of Linear Approximation]\label{lem:V=V_k} \textcolor{black}{For any given LMs $\boldsymbol \gamma$,} 
if $P_{\text{cct}}=0$ and \textcolor{black}{$\epsilon\triangleq\sup_{\{\mathbf{H},\hat{\mathbf{H}}:\mathbf{H}\neq\hat{\mathbf{H}}\}}\Pr\{\mathbf{H}|\hat{\mathbf{H}}\}=0$} (i.e., perfect CSIT), the solution of
the Bellman equation in \eqref{eq:bellman_post} \textcolor{black}{satisfies} 
$V(\widetilde{\mathbf{Q}})=\sum_{k\in \mathcal K}V_{k}(\widetilde{Q}_{k})$ \textcolor{black}{for all $\widetilde{\mathbf{Q}}\in \boldsymbol{\mathcal Q}$},
where $\{V_{k}(\widetilde{Q}_{k})\}$ is the solution of the
per-flow fixed point equation  in \eqref{eq:bellman_perflow}.
~\hfill\IEEEQED
\end{Lem}
\begin{proof}
Please refer to Appendix D.
\end{proof}

\subsection{Distributed Online Learning of  Potential Functions and LMs}
Instead of computing the per-flow potential functions and the
LMs  offline, we  estimate them
distributively at each BS  using  Algorithm \ref{alg:learning}, as illustrated in Steps 1 and 2 in Fig. \ref{fig:learning_diagram}. 

\begin{figure}
 \begin{center}
\includegraphics[height=6.5cm, width=14cm]{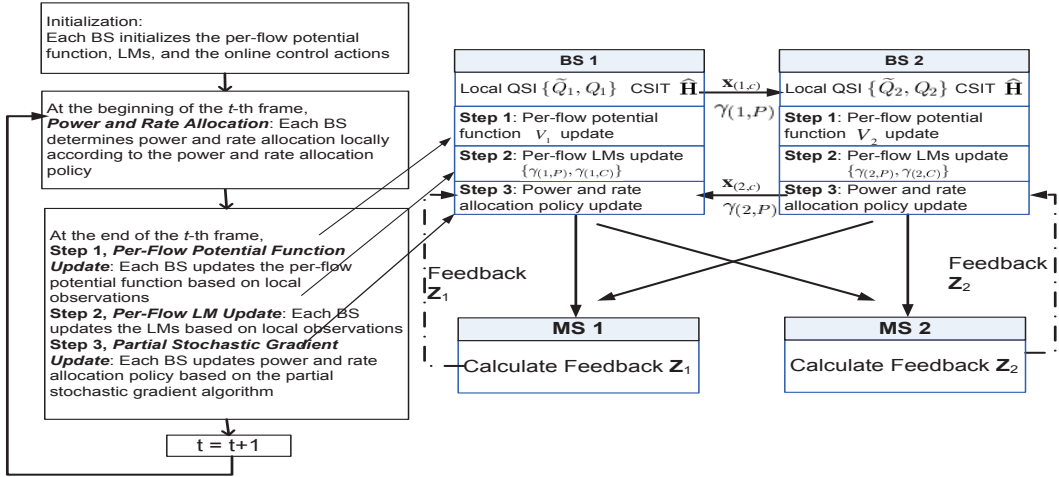}
 \end{center}
    \caption{The system procedure of \textcolor{black}{the proposed low-complexity distributed solution involving Alg. \ref{alg:learning} and Alg. \ref{Alg:gra} with $K=2$.}}
    \label{fig:learning_diagram}
\end{figure}

\begin{Alg}{\em (Distributed Online Learning Algorithm for Per-Flow Potential Functions and LMs)}\label{alg:learning}  At each frame $t$,
\textcolor{black}{let $\widetilde{Q}_k$, $Q_k$  and $\hat{\mathbf{H}}$ be the observed  
post-decision QSI, pre-decision QSI and imperfect CSIT.}  Each BS $k$ updates its per-flow potential functions and LMs according to the following online learning update:
%\begin{equation}
%\label{eq:learn_value_f}\begin{array}{lll}
%V_{k}^{t+1}(\widetilde{Q}_{k})&=
%V_{k}^t(\widetilde{Q}_{k})+\kappa_{v}(t)
%\left[g_k\left(\boldsymbol{\gamma}_{k}^t,\hat{\boldsymbol{\chi}}_k,
%\widetilde{\mathbf P}_k^*,\widetilde{\mathbf R}_k^* \right)+V_{k}^{t}(Q_k-U_{k})
%-V_{k}^{t}(\widetilde{Q}_{k}^0)-V_{k}^{t}(\widetilde{Q}_{k})\right]\\
%%\label{eq:learn_q_factor}
%\gamma_{(k,P)}^{t+1}&=\Gamma\left[\gamma_{(k,P)}^{t}+\kappa_{\gamma}(t)(\textcolor{black}{\widetilde P_{k}^*}-P_k^0)
%\right]\\
%\gamma_{(k,C)}^{t+1}&=\Gamma\left[\gamma_{(k,C)}^{t}+\kappa_{\gamma}(t)(\textcolor{black}{\widetilde R_{k,c}^*}-R_{(k,c)}^0)
%\right]
%\end{array},
%\end{equation}
\begin{align}
\begin{cases}
&V_{k}^{t+1}(\widetilde{Q}_{k})=
V_{k}^t(\widetilde{Q}_{k})+\kappa_{v}(t)
\left[g_k\left(\boldsymbol{\gamma}_{k}^t,\hat{\boldsymbol{\chi}}_k,
\widetilde{\mathbf P}_k^*,\widetilde{\mathbf R}_k^* \right)+V_{k}^{t}(Q_k-U_{k})
-V_{k}^{t}(\widetilde{Q}_{k}^0)-V_{k}^{t}(\widetilde{Q}_{k})\right]\\
&\gamma_{(k,P)}^{t+1}=\Gamma\left[\gamma_{(k,P)}^{t}+\kappa_{\gamma}(t)(\textcolor{black}{\widetilde P_{k}^*}-P_k^0)
\right]\\
&\gamma_{(k,C)}^{t+1}=\Gamma\left[\gamma_{(k,C)}^{t}+\kappa_{\gamma}(t)(\textcolor{black}{\widetilde R_{k,c}^*}-R_{(k,c)}^0)
\right]
\end{cases},\label{eq:learn_value_f}
\end{align}
where $\hat{\boldsymbol{\chi}}_k=\{Q_k,\hat{\mathbf{H}}\}$ and $U_{k}$  is the goodput
to MS $k$ given by \eqref{eq:Ubit_k} under $\hat{\mathbf{H}}=\mathbf{H}$.
\textcolor{black}{$\{\widetilde{\mathbf P}_k^*,\widetilde{\mathbf R}_k^* \}=\widetilde \Omega_{k}^*(\hat{\boldsymbol{\chi}}_k)$. $\widetilde P_{k}^*$ and $\widetilde R_{k,c}^*$} are the power and backhaul consumption of BS $k$ given by $\widetilde \Omega_{k}^*(\hat{\boldsymbol{\chi}}_k)$.
$\Gamma[\cdot]$ is the projection onto an
interval $[0,B]$ for some large constant $B>0$. $\{\kappa_{v}(t)\}$
and $\{\kappa_{\gamma}(t)\}$ are the step size sequences satisfying
the following conditions: 
$\kappa_{v}(t)\geq 0,\sum_t\kappa_{v}(t)=\infty, \  \kappa_{\gamma}(t)\geq 0,
\sum_{t}\kappa_{\gamma}(t) = \infty,\sum_{t}((\kappa_{v}(t))^2+(\kappa_{\gamma}(t))^2)<\infty,
\frac{\kappa_{\gamma}(t)}{\kappa_{v}(t)}\to 0.$~\hfill\IEEEQED
\end{Alg}

\begin{Rem}[Features of  Algorithm \ref{alg:learning}]
Algorithm \ref{alg:learning} only requires local observations of
$\{\widetilde{Q}_{k},Q_k\}$ and $\hat{\mathbf{H}}$ at each BS $k$.
Furthermore, both the per-flow potential \textcolor{black}{functions} and the LMs are updated
simultaneously and distributively at each BS. ~\hfill \IEEEQED
\end{Rem}

In the following, we  establish the convergence proof of \textcolor{black}{Algorithm} \ref{alg:learning}. For given per-flow potential function vector $\mathbf{V}_{k}=\left(V_k(\widetilde Q_k)\right)_{\widetilde Q_k=0,1,\cdots, N_Q}$ and   LMs
$\boldsymbol{\gamma}_{k}$, define a
mapping $T_{k}:\mathbb{R}^{N_Q+1}\to\mathbb{R}$ for the post-decision state $\widetilde{Q}_k$ as follows: $T_{k}(\widetilde{Q}_{k}; \boldsymbol{\gamma}_k,\mathbf{V}_{k})=\text{R.H.S. of \eqref{eq:bellman_perflow}}$. 
%\begin{eqnarray}
%\label{eq:T_{k}}
%T_{k}(\widetilde{Q}_{k}; \boldsymbol{\gamma}_k,\mathbf{V}_{k})=
%\sum\limits_{A_{k}}\Pr\{A_{k}\}\Big\{\min_{{\red \widetilde \Omega_{k}}(Q_k)}\big[
%\widetilde{g}_k(\boldsymbol{\gamma}_{k},Q_{k},{\red \widetilde \Omega_{k}}(Q_{k}))+\sum\nolimits_{\widetilde{Q}_{k}^{\prime}}\textcolor{black}{\Pr\{\widetilde{Q}_{k}^{\prime}|Q_{k},{\red \widetilde \Omega_{k}}(Q_{k})\}}V_{k}(\widetilde{Q}_{k}^{\prime})
%\big]\Big\}
%\end{eqnarray}
\textcolor{black}{Denote $\mathbf{T}_{k}(\boldsymbol{\gamma}_{k}, \mathbf{V}_{k})=\left(T_{k}(\widetilde{Q}_{k};\boldsymbol{\gamma}_k),\mathbf{V}_{k}\right)_{\widetilde Q_k=0,1,\cdots, N_Q}$.}
%The vector form of the mapping $\mathbf{T}_{k}:
%\mathbb{R}^{N_Q+1}\to \mathbb{R}^{N_Q+1}$ is
%given by:
%\begin{equation}
%\label{eq:Tk_vector}
%\mathbf{T}_{k}(\boldsymbol{\gamma}_{k}, \mathbf{V}_{k})=
%\widetilde{\mathbf{g}}_{k}(\boldsymbol{\gamma}_{k})+{\red\mathbb{P}}_{k} \mathbf{V}_k,
%\end{equation}
%where ${\red\mathbb{P}}_{k}$ is a $(N_Q+1)\times(N_Q+1)$ transition matrix
%for the post-decision state queue state of user $k$.
Since we
have two different step size sequences $\{\kappa_{v}(t)\}$ and
$\{\kappa_{\gamma}(t)\}$ with $\kappa_{\gamma}(t) =
o(\kappa_{v}(t))$, the {\em per-flow potential} \textcolor{black}{function} updates and the LM
updates are done simultaneously but over two different timescales. \textcolor{black}{The convergence analysis can be
established over two timescales separately.}
Specifically, during the per-flow potential  function update (timescale I), we have
$\gamma_{(k,C)}^{t+1}-\gamma_{(k,C)}^{t}=O(\kappa_{\gamma}(t))=
o(\kappa_{v}(t))$ and
$\gamma_{(k,P)}^{t+1}-\gamma_{(k,P)}^{t}=O(\kappa_{\gamma}(t))=
o(\kappa_{v}(t))$ for all $k\in \mathcal K$. Therefore, the LMs appear to be
quasi-static  during the per-flow potential function update
in \eqref{eq:learn_value_f}\textcolor{black}{\cite[Chap. 6]{Borkar:2008}}. 
\begin{Lem}[Convergence of Per-flow Potential Function Update (Timescale I)]\label{lem:potential_converge} \textcolor{black}{For given $\boldsymbol{\gamma}_k$},
the iterations of the per-flow potential functions
$ \mathbf{V}_{k}^t$ in   Algorithm
\ref{alg:learning}  converge almost surely to the fixed point of
the per-flow fixed point equation in \eqref{eq:bellman_perflow},
i.e.,
$\lim_{t\rightarrow\infty} \mathbf{V}_{k}^t= \mathbf{V}_{k}^{\infty}$ for all $
k\in \mathcal K$,  where $V_{k}^{\infty}(\widetilde{Q}_{k})$  satisfies:
%is a monotonic increasing function :
\begin{equation}
\label{eq:con_value}
 \mathbf{V}_{k}^{\infty}+V_{k}^{\infty}(\widetilde{Q}_{k}^0)\mathbf{e}=
\mathbf{T}_{k}(\boldsymbol{\gamma}_k, \mathbf{V}_{k}^{\infty}).
\end{equation}
\textcolor{black}{$\mathbf e$ denotes the  $(N_Q+1)$-dimension vector with all-one elements.}
\end{Lem}
\begin{proof}
The proof can be extended from\textcolor{black}{\cite[Chap. 3]{Thesis:Salodkar}} and is omitted due to page limit.
\end{proof}

During the LM update (timescale II), we have
$\lim_{t\to\infty}|V_k^t-V_k^{\infty}(\boldsymbol{\gamma}_k^t)|=0$
w.p.1. for all $
k\in \mathcal K$ \textcolor{black}{\cite[Chap. 6]{Borkar:2008}}. Hence, during the LM update in
\eqref{eq:learn_value_f}, the per-flow potential \textcolor{black}{functions} can be seen as almost
equilibrated. The convergence of the LM update is summarized below.

\begin{Lem}[Convergence of LM Update (Timescale II)]\label{lem:LM_converge}
The iterations of the LMs
$\boldsymbol{\gamma}^t=\{\gamma_{(k,P)}^t,\gamma_{(k,C)}^t:
k\in \mathcal K\}$ in  Algorithm \ref{alg:learning}
converge almost surely to \textcolor{black}{the} invariant set:
\begin{equation}
\mathcal{S}_{\gamma}\triangleq\left\{\boldsymbol{\gamma}:
||\boldsymbol{\gamma}-\boldsymbol{\gamma}^*||^2-\delta_1-
\delta_2\leq0\right\},
\end{equation}
as $t\to\infty$, for some positive \textcolor{black}{constants}
$\delta_1=O(P_\text{cct}^2)$ and $\delta_2=O(\epsilon^2)$, where
$\epsilon=\sup_{\{\mathbf{H},\hat{\mathbf{H}}:\mathbf{H}\neq\hat{\mathbf{H}}\}}\Pr\{\mathbf{H}|\hat{\mathbf{H}}\}$
denotes the CSIT quality.
%$||\hat{\mathbf{H}}-\mathbf{H} ||^2\leq\epsilon^2$.
$\boldsymbol{\gamma}^*=\{\gamma_{(k,P)}^*,\gamma_{(k,C)}^*:k\in \mathcal K\}$ is the
dual optimal solution to the dual problem in \eqref{eqn:L-dual-prob}. 
%that satisfies the \textcolor{black}{average} power and backhaul consumption constraints in \eqref{eq:pwr_con} and \eqref{eq:bkh_con}, respectively. 
~\hfill\IEEEQED
\end{Lem}
\begin{proof}
Please refer to Appendix E.
\end{proof}

%Since the iterations of the LMs will converge almost surely to an
%invariant set finally, there is no loss of optimality by using the
%boundness projection $\Gamma[\cdot]$ in the LMs update.

\subsection{Distributed Online Power and Rate Control via  Stochastic {\red Partial} Gradient Algorithm}

Substituting  \eqref{eq:linear_value} into the R.H.S. of \eqref{eq:bellman_post}, 
the control policy  under linear approximation in \eqref{eq:linear_value} \textcolor{black}{can be obtained} by solving the  following per-stage optimization problem.

\begin{Prob}[Per-Stage Optimization]\label{prob:per_stage} \textcolor{black}{For any given LMs $\boldsymbol \gamma$,} under the linear approximation in \eqref{eq:linear_value},
the online control action (for an observed state realization
$\hat{\boldsymbol{\chi}}$) is given by:
{\red
\begin{equation}\label{eq:per_stage}
\begin{array}{l}
\textcolor{black}{\hat\Omega^*(\hat{\boldsymbol{\chi}})=\{\hat\Omega_P^*(\hat{\boldsymbol{\chi}}),\hat\Omega_R^*(\hat{\boldsymbol{\chi}})\}}
=\argmin\limits_{\mathbf{P},\mathbf{R}}h^{\hat{\boldsymbol{\chi}}}(\mathbf{P},\mathbf{R})
\end{array}
\end{equation}
where
$h^{\hat{\boldsymbol{\chi}}}(\mathbf{P},\mathbf{R})=\sum_{k\in \mathcal K}\mathbb{E}\Big[\gamma_{(k,P)}
\big( P_k^{tx}+P_{\text{cct}}\mathbf{1}( P_k^{tx}>0)\big)+
\gamma_{(k,C)}  R_{(k,c)}+\overline{\mathbf{1}}_{(k,c)}\overline{\mathbf{1}}_{(k,p)}V_k\big(
Q_k \big)
+\mathbf{1}_{(k,c)}\mathbf{1}_{(k,p)}V_k\big( Q_k - 
R_{k}
\big)+\overline{\mathbf{1}}_{(k,c)}\mathbf{1}_{(k,p)}V_k\big(
Q_k - R_{(k,p)})
\big)+\mathbf{1}_{(k,c)}\overline{\mathbf{1}}_{(k,p)}V_k\big(
Q_k - R_{(k,c)} \big)\Big| \hat{\mathbf{H}}\Big]$}. $\mathbf{1}_{(k,c)}=\mathbf{1}\big(R_{(k,c)}\leq
C_{(k,c)}\big)$ and 
$\overline{\mathbf{1}}_{(k,c)}=1-\mathbf{1}_{(k,c)}$. \textcolor{black}{$\mathbf{1}_{(k,p)}$ and $\overline{\mathbf{1}}_{(k,p)}$ are defined in a similar way}.~\hfill\IEEEQED
\end{Prob}

\textcolor{black}{Problem \ref{prob:per_stage} is not tractable as $h^{\hat{\boldsymbol{\chi}}}( \mathbf{P}, \mathbf{R})$ is not  differentiable due to the indicator functions.}
 To solve Problem   \ref{prob:per_stage}, we first use the {\em logistic function}  $f^{\eta}(x,y)=\frac{1}{1+e^{(x-y)\eta}}$  as a smooth approximation for the indicator function $\mathbf 1 (x\leq y)$ in \eqref{eq:per_stage}, \textcolor{black}{ i.e., $
f^{\eta}(x,y)\approx\mathbf{1}(x\leq
y),\forall x,y\in\mathbb{R}^+$ \textcolor{black}{\cite[Chap. 3]{Kenneth:2003}, \cite[Chap. 1]{Kanwal:1998}}}. \textcolor{black}{Note that the approximation is asymptotically accurate  as  $\eta\to \infty$.}  Then, we apply the gradient search method. 
Specifically, the gradient of $h^{\hat{\boldsymbol{\chi}}}(\mathbf{P}, \mathbf{R})$ w.r.t.  a control action
$ a_k\in\{ P_{(k,c)}^i, P_{(k,p)}^i, R_{(k,c)}, R_{(k,p)}:\forall i\}$ of
BS $k$ is given by:
\begin{align}
&\frac{\partial h^{\hat{\boldsymbol{\chi}}}\left(\mathbf{P},\mathbf{R}\right)}{\partial
a_k}\label{eq:gradient}\\
\approx&\mathbb{E}\Big[\underbrace{\frac{\partial
\big[\gamma_{(k,P)} (P_k^{tx}+P_{\text{cct}}f^{\eta}(0, P_k^{tx}) )+
\gamma_{(k,C)} R_{(k,c)} \big] }{\partial  a_k}+\frac{
\partial g_{k}^{\hat{\boldsymbol{\chi}}}(\mathbf{P},\mathbf{R}_k,\mathbf{H},\mathbf V_k)}{\partial a_k}+\textcolor{black}{\sum_{n\in \mathcal K,  n\neq k}}\frac{\gamma_{(n,p)} P_{(k,c),n}}{\partial  a_k} }_{\text{$\triangleq y(a_k)$,  stochastic partial gradient}}\nonumber\\
&+\underbrace{\textcolor{black}{\sum_{n\in \mathcal K,  n\neq k}}\left(\gamma_{(n,p)}
P_{\text{cct}}\frac{\partial f^{\eta}(0, P_n^{tx}) }{\partial
a_k}+\frac{\partial g_{n}^{\hat{\boldsymbol{\chi}}}(\mathbf{P}, \mathbf{R}_n,\mathbf{H},\mathbf V_k)}{\partial a_k}\right)}_{\text{unknown to BS $k$}}\Big| \hat{\mathbf{H}} \Big],\nonumber
\end{align}
\begin{align}
g_{k}^{\hat{\boldsymbol{\chi}}}\left(\mathbf{P}, \mathbf{R}_k,\mathbf{H},\mathbf V_k\right)=&\left(1-f^{\eta}\left( R_{(k,c)},C_{(k,c)}\right)\right)\left(1-f^{\eta}\left( R_{(k,p)},C_{(k,p)}\right)\right)V_k\left(
Q_k
\right)\nonumber\\
&+\left(1-f^{\eta}\left(R_{(k,c)},C_{(k,c)}\right)\right)f^{\eta}\left(R_{(k,p)},C_{(k,p)}\right)V_k\left(
Q_k - R_{(k,p)} \right)\nonumber\\
&
+f^{\eta}\left( R_{(k,c)},C_{(k,c)}\right)\left(1-f^{\eta}\left( R_{(k,p)},C_{(k,p)}\right)\right)V_k\left(
Q_k - R_{(k,c)}
\right)\nonumber\\
&+f^{\eta}\left( R_{(k,c)},C_{(k,c)}\right) f^{\eta}\left( R_{(k,p)},C_{(k,p)}\right)V_k\left( Q_k -
	 R_{k} \right).\nonumber
\end{align}

The gradient $\frac{\partial
h^{\hat{\boldsymbol{\chi}}}\left(\mathbf{P}, \mathbf{R}\right)}{\partial
a_k}$ in \eqref{eq:gradient} cannot be calculated locally at each BS due to the following reasons.  First,
the second term in \eqref{eq:gradient}  is unknown \textcolor{black}{to BS $k$} under the distributed implementation requirement. Second, the expectation $\mathbb{E}$ cannot be computed at BS $k$ without knowledge of the CSIT error kernels under Assumption \ref{ass:csit_model}. In the following, we  propose a \textcolor{black}{distributed}
online stochastic partial  gradient algorithm to obtain the  power and rate control.
\begin{Alg} \textcolor{black}{[Distributed Online  Stochastic Partial Gradient Algorithm for Power and Rate Control]} At each frame $t$, let
$\hat{\boldsymbol{\chi}}_k=\{Q_k,\hat{\mathbf{H}}\}$ denote the observation at each BS $k$.  Each BS $k$ takes control actions $  a_k^t(\hat{\boldsymbol{\chi}}_k)$, \textcolor{black}{obtains $\{\gamma_{(n,P)}^t:n\in \mathcal K, n\neq k\}$ from other BSs through backhaul,}  
and updates  the control  according to the following stochastic partial gradient update:
\begin{align}
& a_k^{t+1}(\hat{\boldsymbol{\chi}}_k)=\left[
 a_k^{t}(\hat{\boldsymbol{\chi}}_k)-
\textcolor{black}{\kappa_a(t)}y\left ( a_k^{t}(\hat{\boldsymbol{\chi}}_k)\right)\right]^+,\label{eqn:update-ak}
\end{align}
where $   a_k\in\{P_{(k,c)}^{i}, P_{(k,p)}^{i}, R_{(k,c)}, R_{(k,p)}:\forall i\}$ and $\{\textcolor{black}{\kappa_a(t)}\}$ is the step size sequence satisfying
the following conditions: $\textcolor{black}{\kappa_a(t)}\geq 0,\sum_t\textcolor{black}{\kappa_a(t)}=\infty, \sum_{t}(\textcolor{black}{\textcolor{black}{\kappa_a(t)}})^2<\infty.$
 \label{Alg:gra}~\hfill\IEEEQED
\end{Alg}
Table \ref{Tab:y-ak} illustrates the detailed  expressions of $y(  a_k^t)$ in \eqref{eqn:update-ak} \textcolor{black}{at frame $t$}.
\begin{table}
\begin{center}
\begin{tabular}{|l|l|}
\hline
  Control Actions $  a_k^t$ & Stochastic Partial Gradient $y(   a_k^t)$ \\
  \hline
$ P_{(k,p)}^{i,t}$ & $\gamma_{(k,P)}^t+ \gamma_{(k,P)}^t
P_{\text{cct}}\frac{\partial f^{\eta}(0, P_k^{tx,t}) }{\partial
P_{(k,p)}^{i,t}}+\frac{
\partial
g_{k}^{\hat{\boldsymbol{\chi}}}(\mathbf{P}^t, \mathbf{R}_k^t,\mathbf{H},\mathbf V_k^t)}{\partial P_{(k,p)}^{i,t}}$\\
$R_{(k,p)}^{t}$ & $\frac{
\partial
g_{k}^{\hat{\boldsymbol{\chi}}}(\mathbf{P}^t,\mathbf{R}_k^t,\mathbf{H},\mathbf V_k^t )}{ \partial  R_{(k,p)}^t}$\\
$P_{(k,c)}^{i,t}$ & $\gamma_{(k,P)}^t\alpha_{(k,k)}^i+\sum_{n\in \mathcal K,  n\neq k}\gamma_{(n,P)}^t\alpha_{(k,n)}^i+\gamma_{(k,P)}^t
P_{\text{cct}}\frac{\partial f^{\eta}(0,P_k^{tx,t}) }{\partial P_{(k,c)}^{i,t}}+ \frac{
\partial
g_{k}^{\hat{\boldsymbol{\chi}}}(\mathbf{P}^t,  \mathbf{R}_k^t,\mathbf{H},\mathbf V_k^t)}{\partial P_{(k,c)}^{i,t}}$ \\
$ R_{(k,c)}^{t}
$ & $\gamma_{(k,C)}^t+ \frac{
\partial
g_{k}^{\hat{\boldsymbol{\chi}}}(\mathbf{P}^t,  \mathbf{R}_k^t,\mathbf{H},\mathbf V_k^t)}{\partial R_{(k,c)}^t}$\\
  \hline
\end{tabular}
  \caption{Expressions of $y(  a_k^t)$ \textcolor{black}{at frame $t$}  for specific control actions.}
  \label{Tab:y-ak}
\end{center}
\end{table}

 The following lemma summarizes the \textcolor{black}{convergence of Algorithm \ref{Alg:gra}}.
\begin{Lem}[Convergence of \textcolor{black}{Algorithm \ref{Alg:gra}}]\label{lem:action} Let
\textcolor{black}{$\hat {\mathcal{W}}^*$} be the set of local minimum points \textcolor{black}{$\hat{\mathbf{W}}^*=\{\hat{\mathbf{P}}^*,\hat{\mathbf{R}}^*\}$ } of  Problem \ref{prob:per_stage}.  The iterations for $\mathbf{W}^{t}=\{\mathbf{P}^t,\mathbf{R}^t\}$ in Algorithm \ref{Alg:gra} converge
almost surely to \textcolor{black}{the} invariant set:
\begin{equation}
\mathcal{S}_w\triangleq \left\{\mathbf{W}:
||\mathbf{W}-\hat{\mathbf{W}}^*||^2-\delta_1-\delta_2 \leq0 \right\},
\end{equation}
as $t\to\infty$, for some \textcolor{black}{local minimum point} \textcolor{black}{$\hat{\mathbf{W}}^*\in\hat{\mathcal{W}}^*$} and
positive constants $\delta_1=O(P_\text{cct}^2)$ and
$\delta_2=O(\epsilon^2)$, where
$\epsilon=\sup_{\{\mathbf{H},\hat{\mathbf{H}}:\mathbf{H}\neq\hat{\mathbf{H}}\}}\Pr\{\mathbf{H}|\hat{\mathbf{H}}\}$
 denotes the CSIT quality. ~\hfill\IEEEQED
\end{Lem}
\begin{proof}
Please refer to Appendix F.
\end{proof}
\textcolor{black}{By Lemma \ref{lem:action},}  Algorithm \ref{Alg:gra} \textcolor{black}{gives the} 
asymptotically local optimal solution at small CSIT errors and $P_{\text{cct}}$.

\textcolor{black}{Finally, we discuss the implementation of Algorithm \ref{Alg:gra} in practical systems.}

\begin{Rem} [\textcolor{black}{Generalized ACK/NAK as MS Feedback for Algorithm \ref{Alg:gra}}]
\textcolor{black}{At each BS $k$},  to compute the stochastic partial gradients in Table \ref{Tab:y-ak}, 
some terms in 
\textcolor{black}{$\frac{
\partial
g_{k}^{\hat{\boldsymbol{\chi}}}(\mathbf{P}^t, \mathbf{R}_k^t,\mathbf{H},\mathbf V_k^t)}{\partial  a_k^t}$ regarding $C_{(k,c)}^t$ and $C_{(k,p)}^t$}  have to be fed back from \textcolor{black}{MS $k$}. Utilizing the property of the logistic function at large $\eta$, the MS feedback can \textcolor{black}{be} substantially simplified. Specifically, for large $\eta$, we have
$$f^{\eta}(x,y))=\frac{1}{1+e^{(x-y)\eta}}\approx\mathbf{1}(x\leq y), \quad J(x-y)\triangleq \frac{-\eta e^{
(x-y)\eta}}{(1+e^{(x-y)\eta})^2}
\approx \frac{\eta}{5}\mathbf{1}
\big(|x-y|\leq \frac{2}{\eta}\big).$$ Note that the approximations are asymptotically accurate  as  $\eta\to \infty$. As a result, we have
\begin{align}
\frac{\partial f^{\eta}(R_{(k,c)},C_{(k,c)})}{\partial
R_{(k,c)}} =&J( R_{(k,c)}-C_{(k,c)}) \approx\frac{\eta}{5}\mathbf{1}
\big(| R_{(k,c)}-C_{(k,c)}|\leq \frac{2}{\eta}\big)\nonumber\\
 \frac{\partial f^{\eta}(R_{(k,c)},C_{(k,c)})}{\partial
 P_{(k,c)}^i}=& J( R_{(k,c)}-C_{(k,c)})
\frac{\sigma_{(k,c)}^i}{(1+\sigma_{(k,c)}^i  P_{(k,c)}^i+I_{(k,c)}^i)\text{ln}2}
\approx\frac{\eta}{5}\mathbf{1}
\big(|R_{(k,c)}-C_{(k,c)}|\leq \frac{2}{\eta}\big)\nonumber
%\\
%\approx&\frac{\eta}{5}\mathbf{1}
%\big(|\hat R_{(k,c)}-C_{(k,c)}|\leq \frac{2}{\eta}\big)\frac{\sigma_{(k,c)}^i}{(1+\sigma_{(k,c)}^i \hat P_{(k,c)}^i+\hat I_{(k,c)}^i)\text{ln}2}\nonumber
\end{align}
Similar notations can be defined for the  private streams with $c$  replaced by $p$.
Based on these
approximations, MS $k$ only needs to feed back a binary vector $$\mathbf{Z}_k=\left\{\mathbf{1}( R_{(k,c)}^t\leq C_{(k,c)}^t), \mathbf{1} \big(|R_{(k,c)}^t-C_{(k,c)}^t|\leq \frac{2}{\eta}\big),\mathbf{1}(R_{(k,p)}^t\leq C_{(k,p)}^t), \mathbf{1} \big(| R_{(k,p)}^t-C_{(k,p)}^t|\leq \frac{2}{\eta}\big) \right\}$$ at each frame $t$  in order for BS $k$ to compute \textcolor{black}{the} stochastic partial gradients. This {\red binary} feedback
has low overhead. \textcolor{black}{In addition,} there are existing built-in \textcolor{black}{mechanisms} in most
wireless systems for these ACK/NAK types of feedback from  MSs.
Furthermore, the convergence property of  Algorithm \ref{Alg:gra}
(Lemma \ref{lem:action}) holds even \textcolor{black}{under} the approximations.~\hfill\IEEEQED
\end{Rem}

\begin{Rem}[\textcolor{black}{Features of  of Algorithm \ref{Alg:gra}}]   \textcolor{black}{Algorithm \ref{Alg:gra} only requires} local  observations
$\{\hat{\boldsymbol{\chi}}_k,\mathbf{Z}_k\}$ and local potential
functions $ \mathbf V_k$ at each BS $k$,  \textcolor{black}{and hence, can be implemented distributively}. In addition, explicit knowledge of the CSIT error kernel  is not required, and hence,  Algorithm \ref{Alg:gra} is robust against uncertainty in the modeling.~\hfill\IEEEQED
\end{Rem}

\section{Simulation and Discussion}\label{sec:simu}
In this section, we  compare the performance of the proposed
distributed solution \textcolor{black}{with} various baseline schemes \textcolor{black}{using numerical simulations}. The average performance is evaluated over $10^6$ iterations. At each frame $t$, we assume the CSI $\mathbf H_{kn}(t)$ is uniformly distributed over a state space $\mathcal H^{N\times M}$ of size $|\mathcal H^{N\times M}|=20$.   We consider Poisson packet arrival with average
arrival rate $\lambda_k$ (packet/s) and exponentially distributed
random packet size with mean $\overline{N}_k=5$Mbits \textcolor{black}{for  $k\in \{1,2\}$}. \textcolor{black}{The buffer
size $N_Q$ is 54Mbits.} The scheduling
frame duration $\tau$ is 5ms. The total BW is $W=10$MHz.  We consider the CSIT error
model  with CSIT error variance $\sigma_e=0.15$\cite{Vincent:MIMO}. The number of transmit and receive antennas is
given by $\{M = 3, N = 2\}$,  the number of common and private
streams for \textcolor{black}{the} Pco-MIMO scheme is $\{d_{(k,c)}=1,d_{(k,p)}=1\}$, $f_k(Q_k)=Q_k$, and $\beta_k=1$  for all $k\in \mathcal K$. \textcolor{black}{We choose $P_1^0=P_2^0$ and $C_1^0=C_2^0$.}

\textcolor{black}{We consider four baseline schemes: Baseline 1 (Coordinative MIMO)\cite{DahroujYu:2010}, Baseline 2 (Uco-MIMO)\cite{MIMO:int:huang}, Baseline 3 (Full Cooperative MIMO)\cite{zhang:coor:2004}, and Baseline 4 (Channel-Aware Pco-MIMO).  Remark \ref{Rem:BL} illustrates the details of the precoder and decorrelator designs in Baselines 1, 2 and 3. Baseline 4 adopts the proposed Pco-MIMO PHY scheme in the precoder and decorrelator design.
All the baseline schemes maximize system throughput under the same backhaul  and power constraints  as the proposed scheme. Therefore, the resulting resource control designs are adaptive to CSIT only, i.e., channel-aware.  Specifically, these four baseline schemes treat the imperfect CSIT as perfect information and do not consider rate allocation due to imperfect CSIT. But Baselines 2, 3, and 4 still consider rate allocation for common streams due to the average backhaul constraints. All the baseline schemes consider power allocation.}

\begin{figure}
\begin{center}
  \subfigure[Average delay per user versus maximum transmit SNR $P_k^0$.  $R_{(k,c)}^0=1.1W\tau\log_2(1+P_k^0)$
(bits/frame), $P_{\text{cct}}=20$dBm, and $\lambda_k=6$ (packet/s).]
  {\resizebox{7.5cm}{!}{\includegraphics{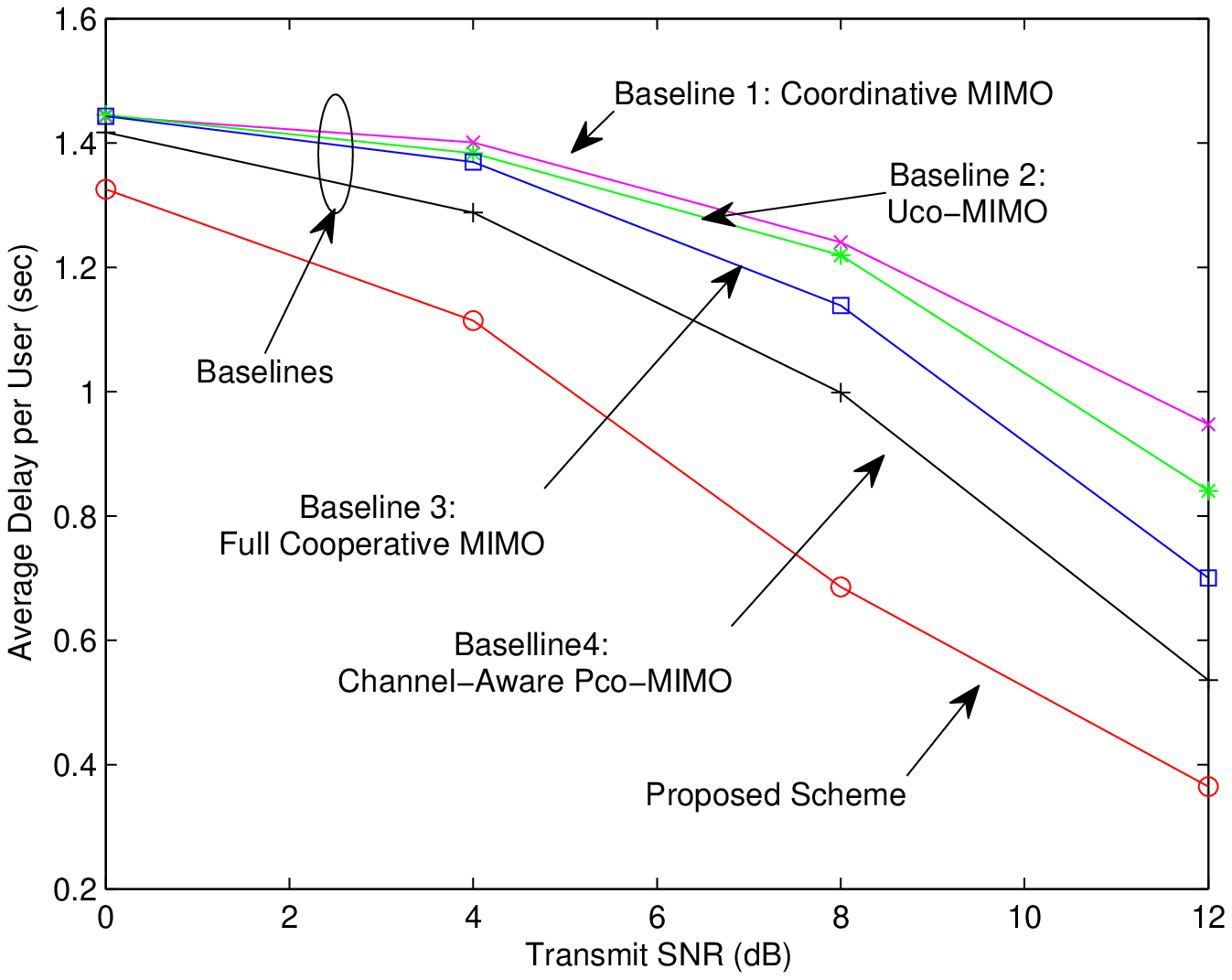}}}\quad \quad 
  \subfigure[Average delay per user versus maximum backhaul consumption $R_{(k,c)}^0$.  $P_k^0=12$dB, $P_{\text{cct}}=20$dBm, and
    $\lambda_k=7$ (packet/s).]
  {\resizebox{7.5cm}{!}{\includegraphics{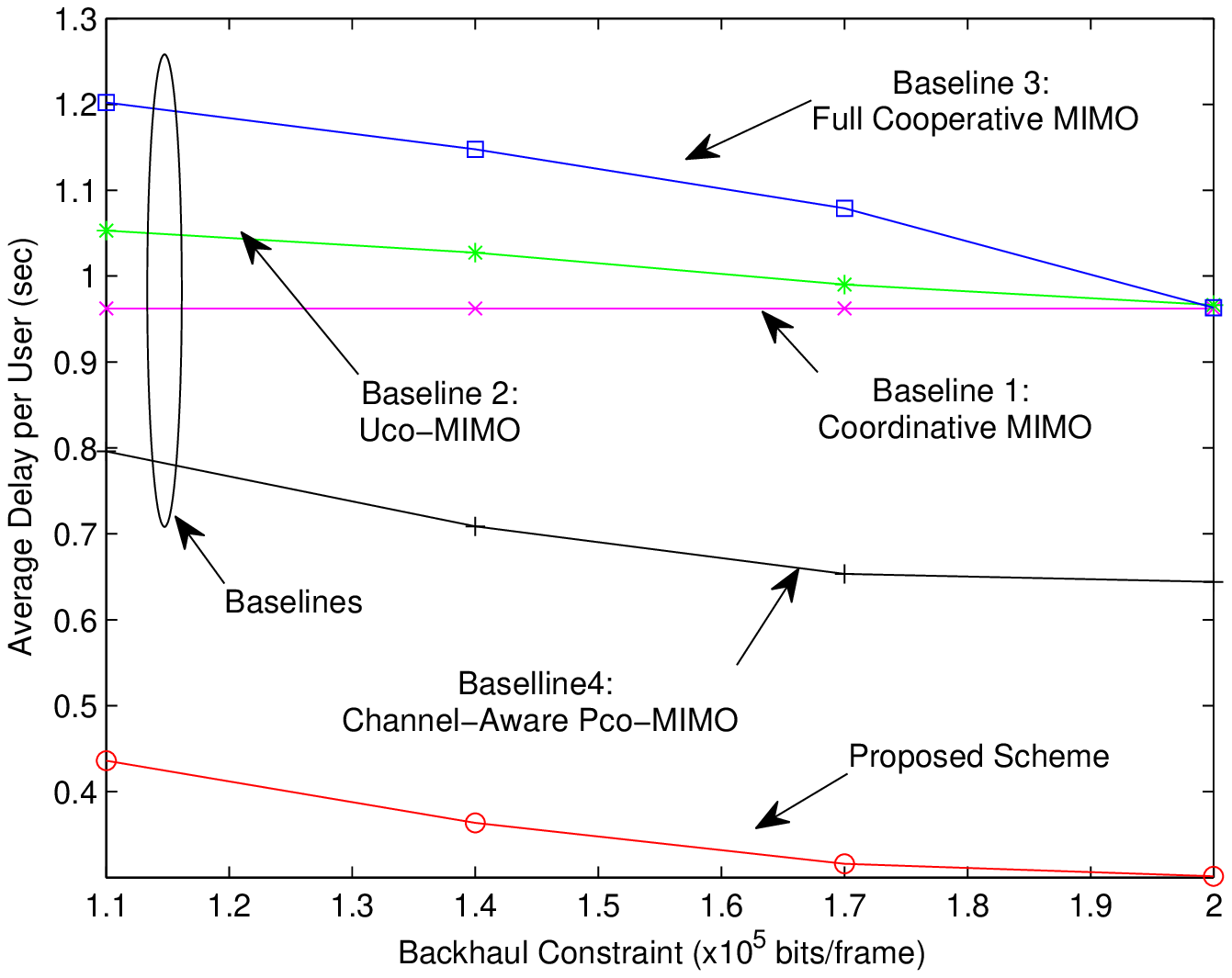}}}
  \end{center}
    \caption{ Average delay per user versus maximum transmit SNR and backhaul consumption.}
    \label{fig:delay}
\end{figure}

\subsection{Delay Performance w.r.t. Transmit SNR}\label{subsec:delay-SNR}

%\begin{figure}
% \begin{center}
%\includegraphics[height=6cm, width=8cm]{figs/delay_pwr_revision.eps}
% \end{center}
%    \caption{\textcolor{black}{\scriptsize{Average delay per user versus maximum transmit SNR $P_k^0$}.  The maximum backhaul consumption  is $R_{(k,c)}^0=1.1W\tau\log_2(1+P_k^0)$
%(bits/frame), and the
%average arrival rate is $\lambda_k=6$ (packet/s).} }
%    \label{fig:delay_pwr_cdf}
%\end{figure}

Fig. \ref{fig:delay} (a) illustrates the average delay per
user versus the maximum transmit SNR  $P_k^0$. The average delay of all the schemes decreases as
the transmit SNR increases.  
%\textcolor{black}{Arranging the baseline schemes according to their performance in decreasing order, we have (Baseline 4, Baseline 3, Baseline 2, Baseline 1).}
\textcolor{black}{This figure demonstrates  the medium backhaul consumption regime, in which Baseline 3 (Full Cooperative MIMO) outperforms Baseline 1 (Coordinative MIMO), while full cooperative MIMO is not the best choice.  The performance gain of Baseline 4 (Channel-Aware Pco-MIMO) compared with Baseline 3 (Full Cooperative MIMO) is contributed by the proposed flexible cooperation level adjustment according to the backhaul consumption requirement. Both Baseline 4 and the proposed scheme apply the proposed Pco-MIMO scheme. The performance gain of the proposed solution compared with Baseline 4 is contributed by the careful delay-aware dynamic power and rate allocation with the consideration of the imperfect CSIT.} It can be seen that the proposed scheme has significant performance
gain  compared with all the baselines.

\subsection{Delay Performance w.r.t. Backhaul Consumption}

%\begin{figure}
% \begin{center}
%\includegraphics[height=6cm, width=8cm]{figs/delay_backhaul_revision.eps}
% \end{center}
%    \caption{\textcolor{black}{\scriptsize{Average delay per user versus maximum backhaul consumption $R_{(k,c)}^0$}.  $P_k^0=12$dB, and
%    $\lambda_k=7$ (packet/s).} }
%    \label{fig:delay_backhaul}
%\end{figure}

Fig. \ref{fig:delay} (b) illustrates the average delay per user
versus the  maximum backhaul consumption $R_{(k,c)}^0$. \textcolor{black}{The average delay of all the schemes decreases as
the backhaul consumption increases.  
%Arranging the baseline schemes according to their performance in decreasing order, we have (Baseline 4, Baseline 1, Baseline 2, Baseline 3). 
This figure demonstrates  the small backhaul consumption regime, in which Baseline 1 (Coordinative MIMO) outperforms Baseline 3 (Full Cooperative MIMO). By carefully making use of the very limited backhaul resources with the proposed flexible cooperation level adjustment, Baseline 4 (Channel-Aware Pco-MIMO) outperforms  Baseline 1. 
%Both Baseline 4 and the proposed scheme apply the proposed Pco-MIMO scheme. The performance gain of the proposed scheme compared with Baseline 4 is contributed by the careful delay-aware dynamic power and rate allocation with the consideration of imperfect CSIT.
In addition, similar comparisons between Baseline 4 and the proposed solution (as in Section \ref{subsec:delay-SNR}) can be made.}  It can be observed that the proposed scheme has significant performance gain  compared with all the baselines.  Note that the delay performance of Baseline 1 is independent of the backhaul
constraint as no data sharing is needed in coordinated
beamforming.

\subsection{Convergence Performance}

\begin{figure}
\begin{center}
  \subfigure[\textcolor{black}{Per User Potential Function}]
  {\resizebox{7.5cm}{!}{\includegraphics{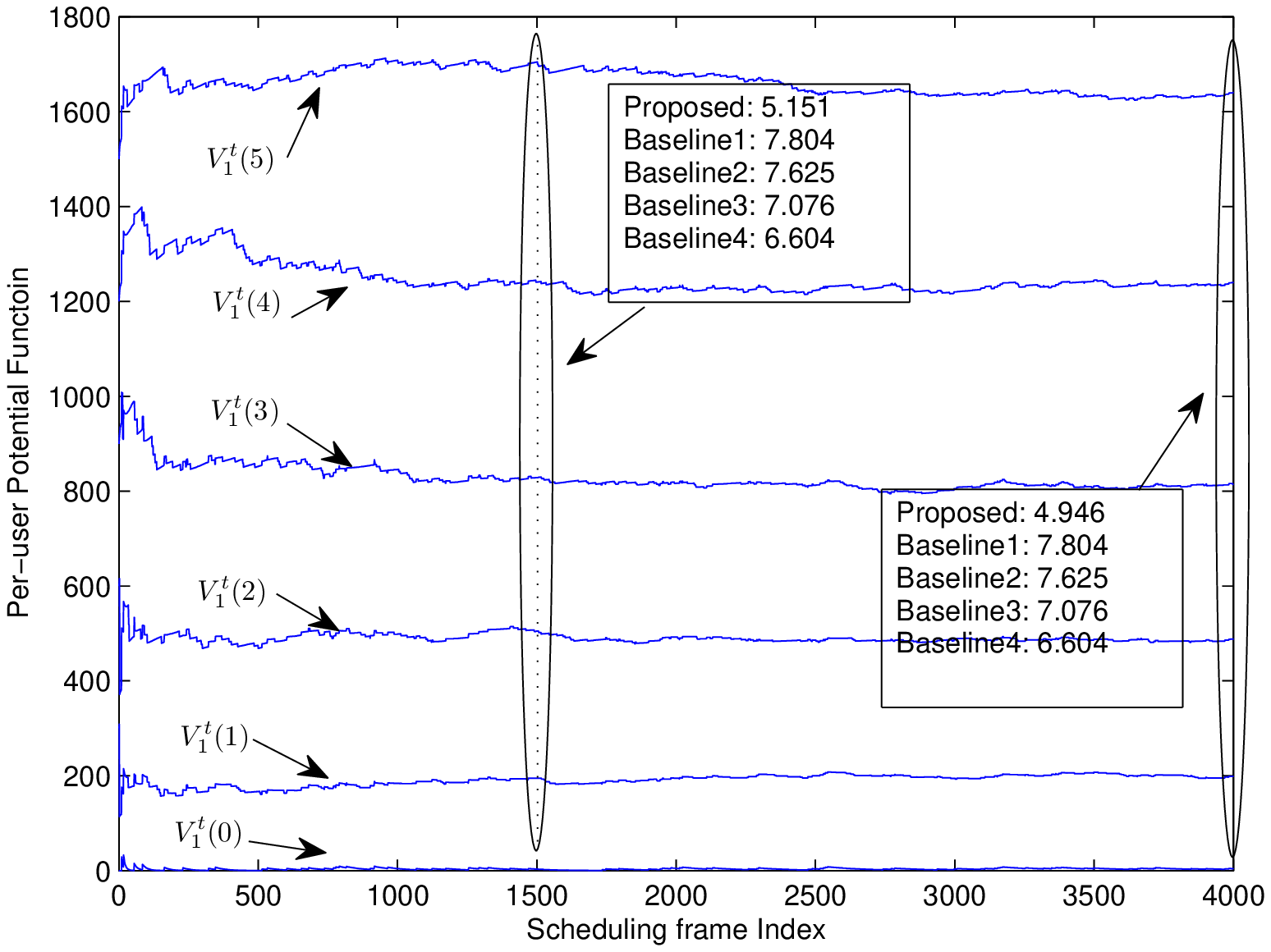}}}
  \subfigure[\textcolor{black}{Power Allocation}]
  {\resizebox{7.5cm}{!}{\includegraphics{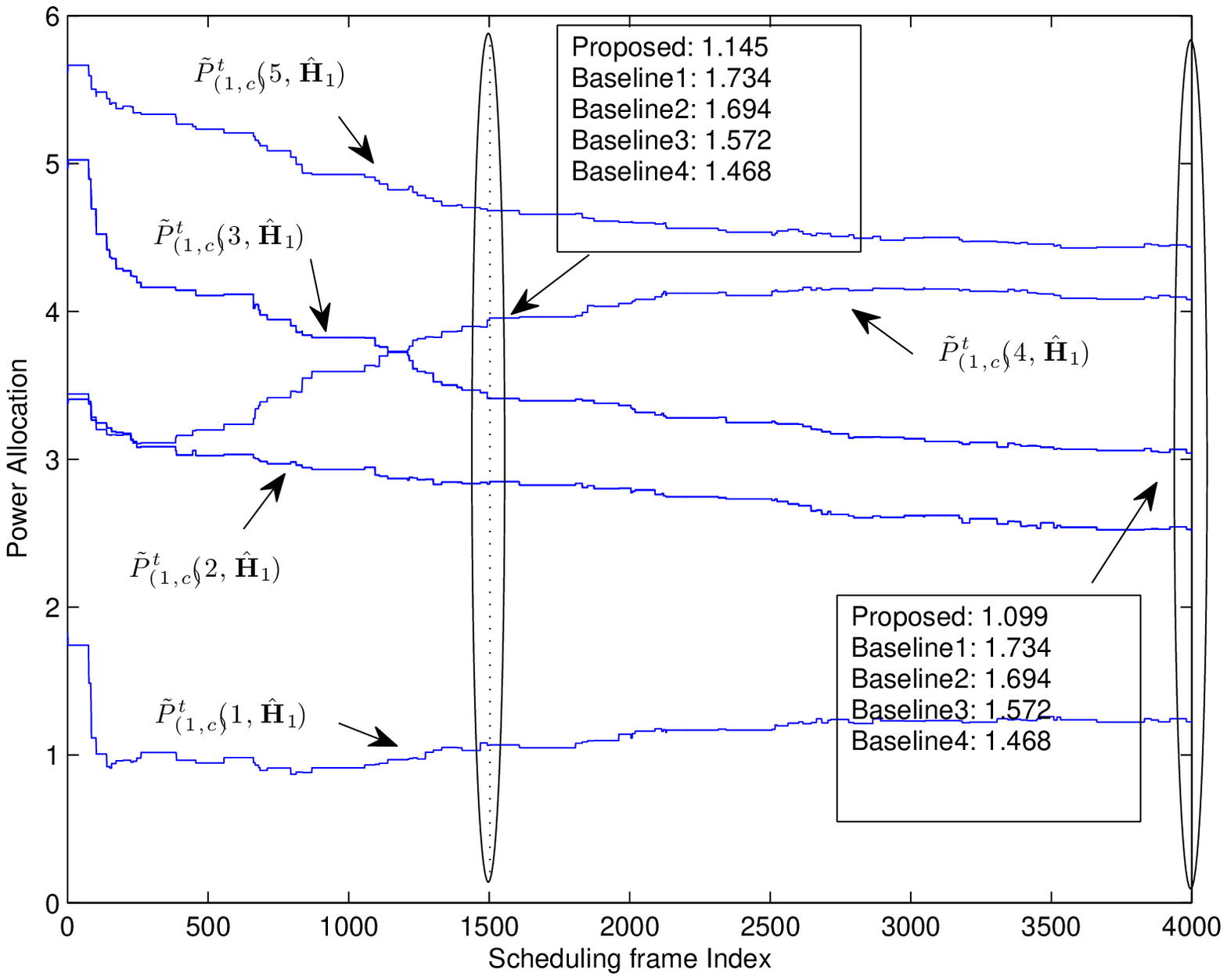}}}
  \end{center}
    \caption{Convergence property of the proposed scheme.  $P_k^0=4$dB, $P_{\text{cct}}=20$dBm, $R_{(k,c)}^0=1.1W\tau\log_2(1+P_k^0)$ (bits/frame),   and  $\lambda_k=4.5$
(packet/s). \textcolor{black}{For illustration, we randomly generate a CSIT realization and plot the power allocation trajectories for different QSI and the randomly generated CSIT realization.}}
    \label{fig:convergence}
\end{figure}

Fig. \ref{fig:convergence} illustrates the convergence property of
the proposed scheme. 
%We plot the instantaneous \textcolor{black}{per-flow} potential
%function $V_1^t$ and power allocation
%$\hat P_{(1,c)}^t(Q_1,\hat{\mathbf{H}}_1)$ versus \textcolor{black}{the} scheduling frame index. 
It can be observed
that the convergence rate of the online algorithm is quite fast. For
example,  the delay performance at 1500-th scheduling frame is
already quite close to the converged average delay.

\textcolor{black}{
\begin{figure}
\begin{center}
  \subfigure[Average delay per user versus transmit SNR $P_k^0$.]
  {\resizebox{7.5cm}{!}{\includegraphics{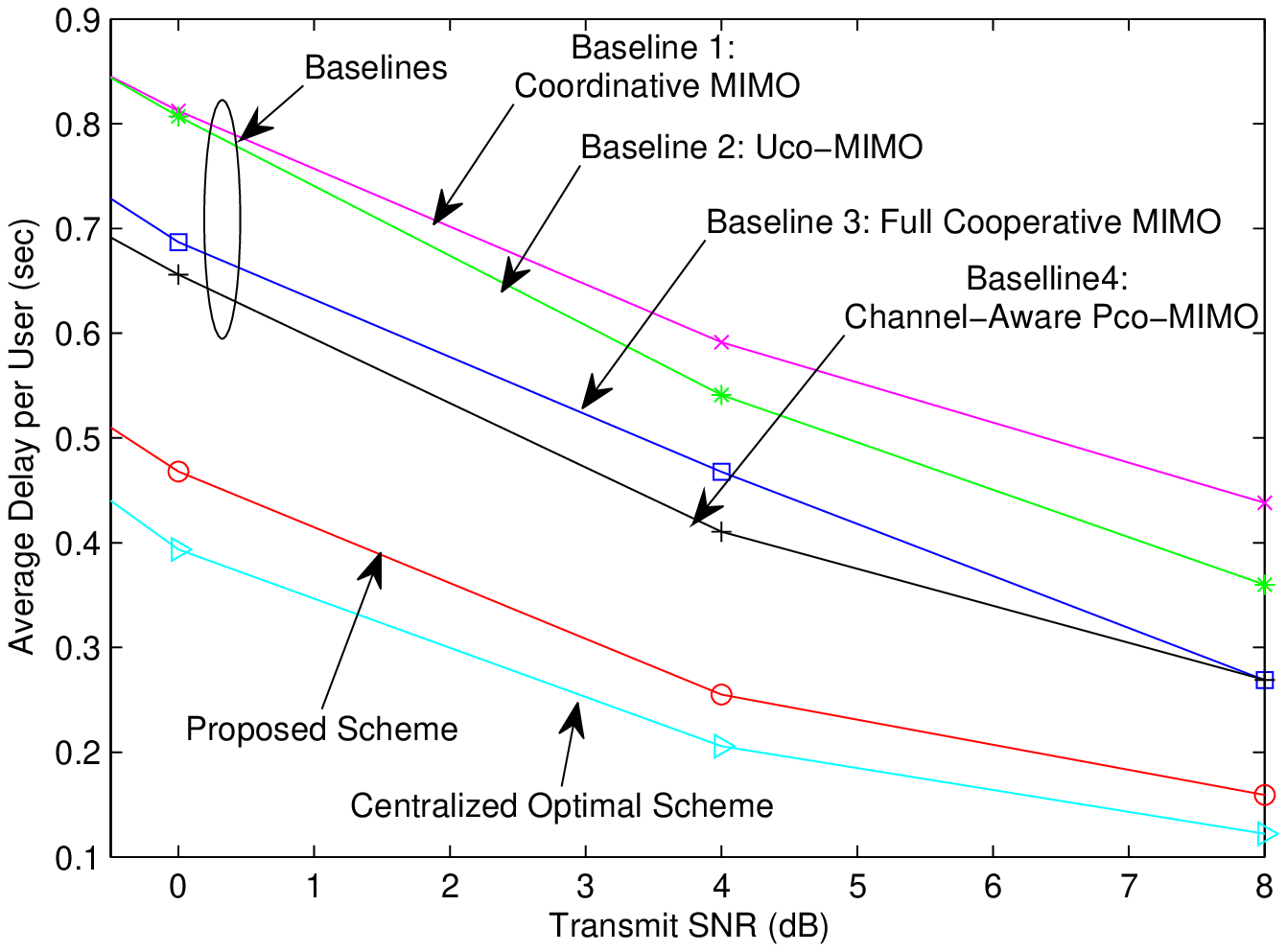}}}
  \subfigure[Matlab computation time per slot and average delay performance per user at $P_k^0=0$dB.]
  {\resizebox{7.5cm}{!}{\includegraphics{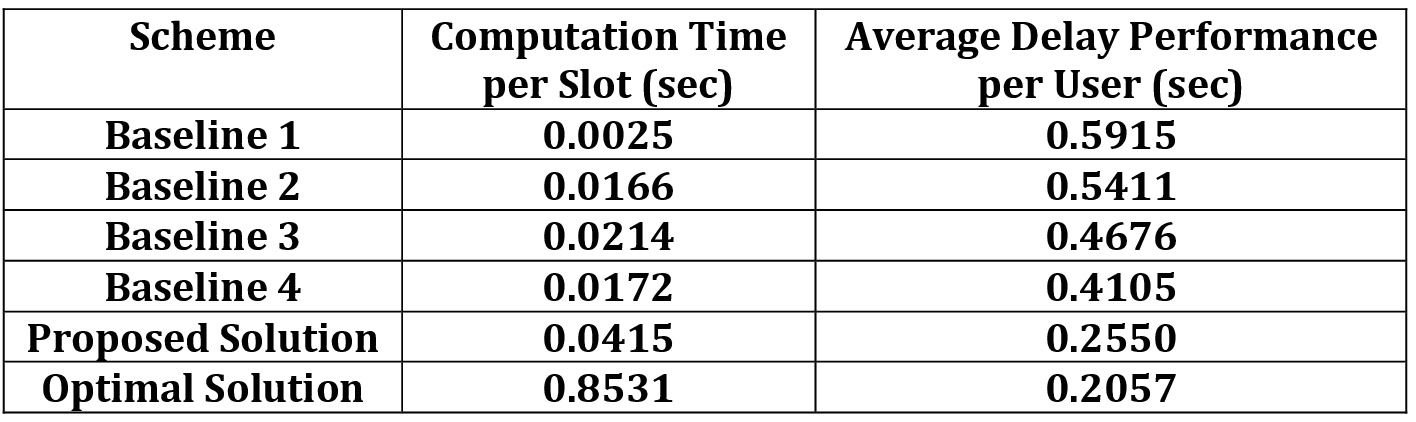}}}
  \end{center}
    \caption{Performance and computational complexity. $R_{(k,c)}^0=W\tau\log_2(1+P_k^0)$ (bits/frame), $P_{\text{cct}}=10$dBm, $N_Q=18$Mbits.}
    \label{fig:complexity}
\end{figure}
\subsection{Computational Complexity}
The computational complexity of the proposed solution is of the same order ($O(K)$) as the four baseline schemes, while it has much lower complexity than the optimal solution  ($O(N_Q^K)$). Fig. \ref{fig:complexity} (a) and Fig. \ref{fig:complexity} (b) illustrate the performance and computational complexity comparisons between the baseline schemes, the proposed distributed solution and the centralized optimal solution.}

\section{Summary}\label{sec:summary}
In this paper, we first propose a novel flexible Pco-MIMO PHY scheme. Based on the Pco-MIMO scheme, we
formulate the delay-optimal   control  problem as an infinite horizon average cost CPOMDP. We
obtain an  equivalent Bellman equation to solve the CPOMDP. \textcolor{black}{To
facilitate  implementation, we propose a low-complexity distributed solution.}  We  prove the convergence and the asymptotical optimality of the proposed solution.

\appendices

\section*{Appendix A: Proof of Theorem~\ref{thm:DoF}}\label{app:DoF}
\textcolor{black}{When CSIT is perfect and $\{d_{(k,c)}, d_{(k,p)}: k\in \mathcal K\}$ satisfies the conditions in \eqref{eq:d_k}, there is no interference under the Pco-MIMO scheme.
Therefore, the system DoF   is given by \textcolor{black}{the total number of non-interfering streams}, i.e., $
\text{DoF}(\text{Pco-MIMO})=\sum_{k\in\mathcal K} (d_{(k,p)}+d_{(k,c)})$.
From
the second constraint in \eqref{eq:d_k}, we have:
\begin{align}
\left(d_{(k,c)}+\sum_{n\in \mathcal K}d_{(n,p)} \right)+\sum_{n\in \mathcal K, n\neq k}d_{(n,c)}\leq N+(K-1)d_{K,M,N} +\sum_{n\in \mathcal K, n\neq k}d_{(n,c)}, \forall k\in\mathcal K, \label{eq:app_tmp1}
\end{align}
where the above equality holds if $d_{(k,c)}+\sum_{n\in \mathcal K}d_{(n,p)} =N+(K-1)d_{K,M,N}$. Since \eqref{eq:app_tmp1} holds for all $k\in \mathcal K$, we can prove \eqref{eqn:DoF-ineq-thm} and  \eqref{eqn:DoF-eq-thm}. Furthermore, by \eqref{eqn:d_kc} for all $k\in \mathcal K$, we can show $\text{DoF}_{\max}(\text{Pco-MIMO})\leq N+(K-1)d_{K,M,N}+(K-1)(N- d_{K,M,N})=KN$, where the equality holds when \eqref{eqn:DoF-eq-max-thm} is satisfied.}

\section*{Appendix B: Proof of Theorem~\ref{thm:MDP_post}} \label{app:MDP_post}
First, we prove Statement (a).  Problem \eqref{eq:dual} can be expressed as
an equivalent MDP:
$\min_{\Omega}\lim\sup_{T\rightarrow\infty}
\frac{1}{T}\sum\nolimits_{t=1}^T \mathbb{E}^{\Omega}\left[\mathbb{E}[g(\boldsymbol{\gamma},\boldsymbol{\chi},\Omega(\hat{\boldsymbol{\chi}}))|\mathbf{Q}]\right]$
with a tuple of the following four objects: the state space
$\{\mathbf{Q}\}$; the action space $\{\Omega(\mathbf{Q})\}$, where
$\Omega(\mathbf{Q})=\{\Omega(\hat{\boldsymbol{\chi}}):\forall \hat{\mathbf H}\}$;
the transition kernel
$\Pr\{\widetilde{\mathbf{Q}}^{\prime}|\mathbf{Q},\Omega(\mathbf{Q})\}=\mathbb{E}
\left[\Pr\{\widetilde{\mathbf{Q}}^{\prime}|\boldsymbol{\chi},\Omega(\hat{\boldsymbol{\chi}})\}|\mathbf{Q}\right]$; and
the per-stage \textcolor{black}{cost} function
$\widetilde{g}(\boldsymbol{\gamma},\mathbf{Q},\Omega(\mathbf{Q}))=
\mathbb{E}[g(\boldsymbol{\gamma},\boldsymbol{\chi},\Omega(\hat{\boldsymbol{\chi}}))|\mathbf{Q}]$.  By standard MDP techniques, we know that the optimal policy $\Omega^*$ can be obtained by solving the equivalent Bellman
equation in \eqref{eq:bellman_post} \textcolor{black}{\cite[Chap. 4]{Bertsekas:2007}}.
Next, we prove Statement (b). By Theorem 2.1 in \cite{Borkaractorcritic:2005}, the support of the randomized policy to the equivalent MDP is included in the set of the optimal solutions: 
$$\arg\min_{\mathbf P, \mathbf R}\mathbb E\left[g(\boldsymbol{\gamma},\boldsymbol{\chi},\mathbf P, \mathbf R)+\sum\nolimits_{\widetilde{\mathbf{Q}}^{\prime}}\Pr\{\widetilde{\mathbf{Q}}^{\prime}|\boldsymbol{\chi},\mathbf P, \mathbf R\}V(\widetilde{\mathbf{Q}}^{\prime})\bigg| \hat{\boldsymbol{\chi}}
\right]$$
Hence, if the set of the optimal solutions is a singleton set \textcolor{black}{for all $\hat{\boldsymbol{\chi}}$}, there is no loss of optimality to focus on deterministic policies. Thus, the deterministic
policy $\Omega^*$ is optimal.

%we can transform the POMDP
%into the MDP with a tuple of the following four objects: \textcolor{black}{the state space
%$\{\mathbf{Q}\}$, the action space $\{\Omega(\mathbf{Q})\}$, where
%$\Omega(\mathbf{Q})=\{\Omega_p(\mathbf{Q}),\Omega_r(\mathbf{Q})\}$,
%the transition kernel
%$\Pr\{\widetilde{\mathbf{Q}}^{\prime}|\mathbf{Q},\Omega(\mathbf{Q})\}=\mathbb{E}
%\left[\Pr\{\widetilde{\mathbf{Q}}^{\prime}|\boldsymbol{\chi},\Omega(\hat{\boldsymbol{\chi}})\}|\mathbf{Q}\right]$ and
%the per-stage \textcolor{black}{cost} function
%$\widetilde{g}(\boldsymbol{\gamma},\mathbf{Q},\Omega(\mathbf{Q}))=
%\mathbb{E}[g(\boldsymbol{\gamma},\boldsymbol{\chi},\Omega(\hat{\boldsymbol{\chi}}))|\mathbf{Q}]$ \cite{Bertsekas:2007,Cao:2007}}.
%By standard MDP techniques, we can derive the equivalent Bellman
%equation in \eqref{eq:bellman_post} \cite{Cao:2007,Thesis:Salodkar}.
%Specifically, $\{\widetilde{V}(\widetilde{\mathbf{Q}})\}$ is the
%potential function. \textcolor{black}{$\Omega^*=\{\Omega_p^*,\Omega_r^*\}$
%minimizing the RHS of \eqref{eq:bellman_post} for all  $\mathbf{Q}\in \mathcal Q$ is the optimal policy},
%i.e.,
%$\Omega^*=\argmin_{\Omega}L_{\beta}(\Omega,\boldsymbol{\gamma})$.
%$\theta=\min_{\Omega}L_{\beta}(\Omega,\boldsymbol{\gamma})$ is the
%optimal average cost per stage.

\section*{Appendix C: Proof of Lemma~\ref{lem:dual_gap}} \label{app:dual_gap}

First, we  show that the duality gap is zero over the stationary randomized policy space.  Define a
stationary randomized policy \textcolor{black}{$\bar{\Omega}$}, which is a
mapping from the observed state $\hat{\boldsymbol{\chi}}$ to some
measurable $f:\mathcal{S}\to\mathcal{P}(\mathcal{U})$, where
$\mathcal{S}$ is the observed state space, $\mathcal{U}$ is the
power and rate allocation space, and $\mathcal{P}$ is the Polish
space of probability measures on $\mathcal{U}$ with the Prohorov
topology\textcolor{black}{\cite[Chap. 2]{borkar:probability:1995}}.
The observed state
$\hat{\boldsymbol{\chi}}=\{\mathbf{Q},\hat{\mathbf{H}}\}$ under the randem 
control of the unichain policy \textcolor{black}{$\bar{\Omega}$} has an invariant
probability measure $\pi\in\mathcal{P}(\mathcal{S})$. The ergodic
occupation measure $\pi_{\textcolor{black}{\bar{\Omega}}}$ associated with the
pair $(\pi,\textcolor{black}{\bar{\Omega}})$ is defined
by\cite{Borkar::convex:2001}: 
\begin{equation}
\int_{\mathcal{S}\times\mathcal{U}}g(\hat{\boldsymbol{\chi}},y)\text{d}\pi_{\textcolor{black}{\bar{\Omega}}}
=\sum\nolimits_{\hat{\boldsymbol{\chi}}\in\mathcal{S}}\pi(\mathcal{S})
\int_{\mathcal{U}}
g(\hat{\boldsymbol{\chi}},y)f(\hat{\boldsymbol{\chi}},\text{d}y),
\end{equation}
where $g(\hat{\boldsymbol{\chi}},y)$ is the per stage \textcolor{black}{cost}
function given the observed state is $\hat{\boldsymbol{\chi}}$ and
action $y$ is taken. Let $\mathcal{G}$ denote the set of all ergodic
occupation measures $\pi_{\textcolor{black}{\bar{\Omega}}}$, and it has been shown in
\cite{Borkar::convex:2001} that $\mathcal{G}$ is closed convex in
$\mathcal{P}(\mathcal{S}\times\mathcal{U})$. Therefore, the primal
Problem \ref{prob:delay} can be recast as a convex problem given by:
\begin{equation}\label{eq:xxx}\begin{array}{lll}
\min_{\nu\in\mathcal{G}} & \int g(\hat{\boldsymbol{\chi}},y)\text{d}\nu\\
\text{s.t.} & \int P_k(\hat{\boldsymbol{\chi}},y)\text{d}\nu\leq
P_k^0;\quad \int R_{(k,c)}(\hat{\boldsymbol{\chi}},y)\text{d}\nu\leq
R_{(k,c)}^0,\forall k
\end{array},
\end{equation}
which is an infinite dimensional linear program
\cite{Borkar::convex:2001}, \textcolor{black}{\cite[Chap. 1]{LP:infinite:1987}}. Define the
Lagrangian function: $L_{\text{LP}}(\nu,\boldsymbol{\gamma})=\int
g\text{d}\nu+\sum_k\gamma_{(k,P)}(\int P_k\text{d}\nu-P_k^0)+
\sum_k\gamma_{(k,C)}(\int R_{(k,c)}\text{d}\nu-R_{(k,c)}^0).$ \textcolor{black}{We have the saddle-point condition: $L_{\text{LP}}(\nu,\boldsymbol{\gamma}^*) \geq L_{\text{LP}}(\nu^*,\boldsymbol{\gamma}^*)\geq L_{\text{LP}}(\nu^*,\boldsymbol{\gamma})$, {\red i.e., the duality gap is zero over the stationary randomized policy space.}}
%Next, from Theorem \ref{thm:MDP_post}, we know that there exists an optimal deterministic policy  given by the solution of the equivalent Bellman equation. In other words, there is no loss of optimality by focusing on the stationary deterministic policy space. } Hence, we have $\min_{\nu}L_{\text{LP}}(\nu,\boldsymbol{\gamma})=
%\min_{\Omega}L_{\beta}(\Omega,\boldsymbol{\gamma})$ for any $\boldsymbol \gamma$. In other words,
%the ergodic occupation measure $\pi_{\Omega^*}=\nu^*$, where
%$\Omega^*=\arg\min_{\Omega}L_{\beta}(\Omega,\boldsymbol{\gamma})$
%and $\nu^*=\arg\min_{\nu}L_{\text{LP}}(\nu,\boldsymbol{\gamma})$. \textcolor{black}{Therefore, based on the zero duality gap on over the stationary randomized policy space, we have the zero duality gap  in
%\eqref{eq:zero_duality} over the stationary deterministic policy space.

{\black Next, from Theorem \ref{thm:MDP_post} (b), there is no loss of optimality by focusing on deterministic policies {\red given that the condition of Theorem \ref{thm:MDP_post} (b) holds. } Hence, we have $\min_{\nu}L_{\text{LP}}(\nu,\boldsymbol{\gamma})=
\min_{\Omega}L_{\beta}(\Omega,\boldsymbol{\gamma})$ for any $\boldsymbol \gamma \succeq 0$. As a result, the saddle point condition holds for the constrained Problem \ref{prob:delay} over the domain of deterministic policies, i.e., $L_{\beta}(\Omega, \gamma^*) \geq L_{\beta}(\Omega^*, \gamma^*) \geq L_{\beta}(\Omega^*, \gamma)$ for all deterministic policies $\Omega$ and $\boldsymbol\gamma \succeq 0$. As a result, $(\Omega^*,\gamma^*)$ is the saddle point of $L_{\beta}(\Omega, \gamma   )$ and the duality gap is zero, i.e., \eqref{eq:zero_duality} holds.}

\section*{Appendix D: Proof of Lemma~\ref{lem:V=V_k}} \label{app:V=V_k}
\textcolor{black}{When $\hat{\mathbf{H}}=\mathbf{H}$, there is no interference under the Pco-MIMO. Thus, 
given 
$\hat{\boldsymbol{\chi}}=\boldsymbol{\chi}$ and $\{\mathbf P, \mathbf R\}$,  
$\widetilde{Q}^{\prime}_k=Q_k-U_k(\hat{\mathbf{H}},\mathbf P_k, \mathbf R_k )$ is independent of  $Q_n$ and $\{\mathbf P_n, \mathbf R_n\}$ for all $n\in \mathcal K, n\neq k$. Thus, we have
$\Pr\{\widetilde{Q}_k^{\prime}|\hat{\boldsymbol{\chi}},\mathbf P, \mathbf R\}=\Pr\{\widetilde{Q}_k^{\prime}|
\hat{\boldsymbol{\chi}}_k,\mathbf P_k, \mathbf R_k\}.$ 
When $P_{\text{cct}}=0$, we have
$g(\boldsymbol{\gamma},\hat{\boldsymbol{\chi}},\mathbf P, \mathbf R)=\sum_k
g_k(\boldsymbol{\gamma}_k,\hat{\boldsymbol{\chi}}_k,\mathbf P_k, \mathbf R_k).$ Suppose
$V(\widetilde{\mathbf{Q}})=\sum_{k}V(\widetilde Q_k)$
and $\theta=\sum_k V_k(\widetilde{Q}_k^0)$. Then, the Bellman
equation in \eqref{eq:bellman_post} becomes: 
\begin{align}
&\sum_kV_k(\widetilde{Q}_k)+\sum_kV_k(\widetilde{Q}_k^0)\nonumber\\
=&\mathbb E\left[ 
\min_{\mathbf P, \mathbf R}\left[g(\boldsymbol{\gamma},\boldsymbol{\chi},\mathbf P, \mathbf R)+\sum\nolimits_{\widetilde{\mathbf{Q}}^{\prime}}\Pr\{\widetilde{\mathbf{Q}}^{\prime}|\boldsymbol{\chi},\mathbf P, \mathbf R\}\left(\sum_kV_k(\widetilde{Q}_k^{\prime})\right)\right]\Bigg| \widetilde{\mathbf Q}\right]\nonumber\\
\stackrel{(a)}{=}&\mathbb E\left[
\min_{\mathbf P, \mathbf R}\left[\sum_k\left(
g_k(\boldsymbol{\gamma}_k,\hat{\boldsymbol{\chi}}_k,\mathbf P_k, \mathbf R_k)+\sum_{\widetilde{Q}_k^{\prime}}
\Pr\{\widetilde{Q}_k^{\prime}|\hat{\boldsymbol{\chi}}_k, \mathbf P_k, \mathbf R_k\}V_k(\widetilde{Q}_k^{\prime})\right)
\right]\Bigg|\widetilde{\mathbf Q}\right]\nonumber\\
%=&\sum_{\mathbf{A}}\Pr\{\mathbf{A}\}\sum_k\mathbb E\left[
%\min_{\mathbf P_k, \mathbf R_k}\left[
%g_k(\boldsymbol{\gamma}_k,\hat{\boldsymbol{\chi}}_k,\mathbf P_k, \mathbf R_k)+\sum_{\widetilde{Q}_k^{\prime}}
%\Pr\{\widetilde{Q}_k^{\prime}|\hat{\boldsymbol{\chi}}_k, \mathbf P_k, \mathbf R_k\}V_k(\widetilde{Q}_k^{\prime})
%\right]\Bigg|Q_k\right]\nonumber\\
\stackrel{(b)}{=}&\sum_k\mathbb E\left[
\min_{\mathbf P_k, \mathbf R_k}\left[
g_k(\boldsymbol{\gamma}_k,\hat{\boldsymbol{\chi}}_k,\mathbf P_k, \mathbf R_k)+\sum_{\widetilde{Q}_k^{\prime}}
\Pr\{\widetilde{Q}_k^{\prime}|\hat{\boldsymbol{\chi}}_k,\mathbf P_k, \mathbf R_k\}V_k(\widetilde{Q}_k^{\prime})
\right]\Bigg|\widetilde Q_k\right],\label{eq:bell_tmp}
%=&\sum_{\mathbf{A}}\Pr\{\mathbf{A}\}\left\{
%\min_{\Omega(\mathbf{Q})}\left[\sum_k
%\widetilde{g}_k(\boldsymbol{\gamma}_k,Q_k,\Omega_k(\mathbf Q))+\sum\nolimits_{\widetilde{\mathbf{Q}}^{\prime}}\Pr\{\widetilde{\mathbf{Q}}^{\prime}|\mathbf{Q},\Omega(\mathbf{Q})\}\left(\sum_kV_k(\widetilde{Q}_k^{\prime})\right)
%\right]\right\}\nonumber\\
%\stackrel{(a)}{=}&\sum_{\mathbf{A}}\Pr\{\mathbf{A}\}\left\{
%\min_{\Omega(\mathbf{Q})}\left[\sum_k\left(
%\widetilde{g}_k(\boldsymbol{\gamma}_k,Q_k,\Omega_k(\mathbf Q))+\sum_{\widetilde{Q}_k^{\prime}}
%\Pr\{\widetilde{Q}_k^{\prime}|Q_k, \Omega_k(\mathbf{Q})\}V_k(\widetilde{Q}_k^{\prime})\right)
%\right]\right\}\nonumber\\
%=&\sum_{\mathbf{A}}\Pr\{\mathbf{A}\}\left\{\sum_k
%\min_{\widetilde\Omega_k(Q_k)}\left[
%\widetilde{g}_k(\boldsymbol{\gamma}_k,Q_k,\widetilde \Omega_k(Q_k))+\sum_{\widetilde{Q}_k^{\prime}}
%\Pr\{\widetilde{Q}_k^{\prime}|Q_k, \widetilde\Omega(Q_k)\}V_k(\widetilde{Q}_k^{\prime})
%\right]\right\}\nonumber\\
%\stackrel{(b)}{=}&\sum_k\sum_{A_k}\Pr\{A_k\}\left\{
%\min_{\widetilde\Omega_k(Q_k)}\left[
%\widetilde{g}_k(\boldsymbol{\gamma}_k,Q_k,\widetilde \Omega_k(Q_k))+\sum_{\widetilde{Q}_k^{\prime}}
%\Pr\{\widetilde{Q}_k^{\prime}|Q_k, \widetilde\Omega(Q_k)\}V_k(\widetilde{Q}_k^{\prime})
%\right]\right\}\label{eq:bell_tmp}
\end{align}
where (a) is due to $\sum\nolimits_{\widetilde{\mathbf{Q}}^{\prime}}\Pr\{\widetilde{\mathbf{Q}}^{\prime}|\hat{\boldsymbol{\chi}},\mathbf P, \mathbf R\}\left(\sum_kV_k(\widetilde{Q}_k^{\prime})\right)
=\sum_k \sum\nolimits_{\widetilde{\mathbf{Q}}^{\prime}}\Pr\{\widetilde{\mathbf{Q}}^{\prime}|\hat{\boldsymbol{\chi}},\mathbf P, \mathbf R\}V_k(\widetilde{Q}_k^{\prime})=\sum_k \sum\nolimits_{\widetilde{Q}_k^{\prime}}\Pr\{\widetilde{Q}_k^{\prime}|\hat{\boldsymbol{\chi}},\mathbf P, \mathbf R\}V_k(\widetilde{Q}_k^{\prime})=\sum_k \sum\nolimits_{\widetilde{Q}_k^{\prime}}\Pr\{\widetilde{Q}_k^{\prime}|\hat{\boldsymbol{\chi}}_k,\mathbf P_k, \mathbf R_k\}V_k(\widetilde{Q}_k^{\prime})$ and (b) is due to the independent assumptions w.r.t. $k$ in the CSI model, imperfect CSI model and bursty source model. Therefore,  \eqref{eq:bell_tmp} can be recast into
per-flow Bellman equations given by
 \eqref{eq:bellman_perflow} for each MS $k$.} Furthermore, since the
solution of the Bellman equation is unique up to a constant, we can
conclude that when $V_k(\widetilde Q_k)$ is a solution to the per-flow
fixed point equation in \eqref{eq:bellman_perflow},
$V(\widetilde{\mathbf{Q}})=\sum_{k}V_k(\widetilde Q_k)$
is a solution of the Bellman equation in \eqref{eq:bellman_perflow}.

\section*{Appendix E: Proof of Lemma~\ref{lem:LM_converge}} \label{app:LM_converge}
\textcolor{black}{First, we obtain the ordinary differential equation (ODE) of the LM update in Algorithm \ref{alg:learning}.} Due to the separation of \textcolor{black}{timescales}, the primal update of the
potential function can be regarded as converged to
\textcolor{black}{$\mathbf V_k(\boldsymbol{\gamma}^t)$} w.r.t. the current LMs
$\boldsymbol{\gamma}^t$. \textcolor{black}{Let $\widetilde \Omega^*=\{\widetilde \Omega_k^*:k\in \mathcal K\}$.}
Using the standard stochastic approximation
argument {\red in Lemma 1 of} \textcolor{black}{\cite[Chap. 6]{Borkar:2008}}, the dynamics of the LMs learning
equation {\red under $\widetilde \Omega^*(\boldsymbol{\gamma}(t))$)} can be represented by the
following ODE:
\textcolor{black}{
\begin{equation}\label{eq:gamma_ODE}
\dot{\boldsymbol{\gamma}}(t)=f_{\widetilde\Omega^*}(\boldsymbol{\gamma}(t))\textcolor{black}{\triangleq}\mathbb{E}^{\widetilde\Omega^*(\boldsymbol{\gamma}(t))}
[\widetilde P_{1}^*(\boldsymbol{\gamma}(t))-P_1^0,\widetilde P_{2}^*(\boldsymbol{\gamma}(t))-P_2^0,\widetilde R_{(1,c)}^*(\boldsymbol{\gamma}(t))-C_1^0,\widetilde R_{(2,c)}^*(\boldsymbol{\gamma}(t))-C_2^0]^T,
\end{equation}
where $\widetilde P_{k}^*(\boldsymbol{\gamma}(t))$ and $\widetilde R_{(k,c)}^*(\boldsymbol{\gamma}(t))$ \textcolor{black}{are} the power
and backhaul consumption given by $\widetilde\Omega^*(\boldsymbol{\gamma}(t),\hat{\boldsymbol{\chi}})$.}
At the equilibrium point \textcolor{black}{$\boldsymbol{\gamma}^*$} of the ODE
\eqref{eq:gamma_ODE}, we have \textcolor{black}{$\boldsymbol{\gamma}^*\cdot f_{\widetilde \Omega^*}(\boldsymbol{\gamma}^*)=0$},
which satisfies the power and backhaul consumption constraints in
\eqref{eq:pwr_con} and \eqref{eq:bkh_con} \textcolor{black}{(by KKT conditions)}.

Next, we consider the case when $\epsilon=0$ and $P_{\text{cct}}=0$.  By Lemma \ref{lem:V=V_k}, \textcolor{black}{we have $\Omega^*=\widetilde \Omega^*$ when $\epsilon=0$ and $P_{\text{cct}}=0$}. The ODE in \eqref{eq:gamma_ODE} becomes:
\begin{equation}\label{eq:ODE_fake}
\dot{\boldsymbol{\gamma}}(t)=f_{\Omega^*}(\boldsymbol{\gamma}(t))\textcolor{black}{\triangleq}\mathbb{E}^{\Omega^*(\boldsymbol{\gamma}(t))}
[P^*_{1}(\boldsymbol{\gamma}(t))-P_1^0, P^*_{2}(\boldsymbol{\gamma}(t))-P_2^0,R^*_{(1,c)}(\boldsymbol{\gamma}(t))-C_1^0,R^*_{(2,c)}(\boldsymbol{\gamma}(t))-C_2^0]^T,
\end{equation}
where $P_{k}^*(\boldsymbol{\gamma}(t))$ and $ R_{(k,c)}^*(\boldsymbol{\gamma}(t))$ is the power
and backhaul consumption \textcolor{black}{given by $\Omega^*(\boldsymbol{\gamma}(t),\hat{\boldsymbol{\chi}})$.}
\textcolor{black}{On the other hand, since $\Omega^*(\boldsymbol{\gamma})$ is optimal and $\epsilon=0$ and $P_{\text{cct}}=0$}, we have
$\Omega^*(\boldsymbol{\gamma})=\arg\min_{\Omega(\boldsymbol{\gamma})}\mathbb{E}^{\Omega(\boldsymbol{\gamma})}[\sum_k g_k(\boldsymbol{\gamma}_{k},\hat{\boldsymbol{\chi}}_k,\Omega_{k}(\boldsymbol \gamma,\hat{\boldsymbol{\chi}}_k))
]$. Define
$G(\boldsymbol{\gamma})=\mathbb{E}^{\Omega^*(\boldsymbol{\gamma})}[
\sum_k g_k(\boldsymbol{\gamma}_{k},\hat{\boldsymbol{\chi}},\Omega_{k}^*(\boldsymbol \gamma,\hat{\boldsymbol{\chi}}))\textcolor{black}{]}$. 
By the envelope theorem, we have $\frac{\partial
G(\boldsymbol{\gamma})}{\partial \gamma_{(k,P)}
}=\mathbb{E}^{\Omega^*(\boldsymbol{\gamma})}[
P^*_k(\boldsymbol{\gamma})-P_k^0]=\dot{\gamma}_{(k,P)}.$ Similarly, we have
$\frac{\partial G(\boldsymbol{\gamma})}{\partial \gamma_{(k,C)}
}=\mathbb{E}^{\Omega^*(\boldsymbol{\gamma})}[
R^*_{(k,c)}(\boldsymbol{\gamma})-R_{(k,c)}^0]=\dot{\gamma}_{(k,C)}.$ Therefore, we
can show that the ODE in \eqref{eq:ODE_fake} can be expressed as
$\dot{\boldsymbol \gamma}(t)= \triangledown
G(\boldsymbol{\gamma}(t))$.  Since the dual function $G(\boldsymbol{\gamma})$ is a concave function, from the standard gradient update argument, the ODE in \eqref{eq:ODE_fake} will converge to the equilibrium point \textcolor{black}{$\boldsymbol{\gamma}_0^*$}. Thus,  we have $\boldsymbol{\gamma}_0^*\cdot\triangledown
G(\boldsymbol{\gamma}_0^*)=0$.  $\boldsymbol{\gamma}_0^*$ corresponds to the LMs associated with the power and backhaul constraints under the optimal policy (by KKT conditions). Furthermore, the equilibrium point $\boldsymbol{\gamma}_0^*$  is
exponentially stable on $\mathcal{R}^+$. By \textcolor{black}{the} {\red convergence} of \textcolor{black}{the} Lyapunov
Theorem \cite{Lyapunov:1998}, there exists a Lyapunov function
${\red \mathcal L}(\boldsymbol{\gamma})$ for \textcolor{black}{$\dot{\boldsymbol{\gamma}}(t)=
f_{\Omega^*}(\boldsymbol{\gamma}(t))$}, s.t.
$C_1||\boldsymbol{\gamma}-\boldsymbol{\gamma}_0^*||^2\leq
{\red \mathcal L}(\boldsymbol{\gamma})\leq
C_2||\boldsymbol{\gamma}-\boldsymbol{\gamma}_0^*||^2 $ and 
$\frac{\text{d} {\red \mathcal L}(\boldsymbol{\gamma})}{\text{d}\boldsymbol{\gamma}}
f_{\widetilde{\Omega}}(\boldsymbol{\gamma})\leq-C_3||\boldsymbol{\gamma}-\boldsymbol{\gamma}_0^*||^2$ for all $
\boldsymbol{\gamma}\in\mathcal{R}^+ $ and for some positive constant
$\{C_1,C_2,C_3\}$.

Finally, consider general  $\epsilon$ and $P_{\text{cct}}$.
Using the standard perturbation analysis, we have
\textcolor{black}{\begin{align}
&\boldsymbol{\phi}_{(k,P)}(\boldsymbol{\gamma})\triangleq\mathbb{E}^{\widetilde{\Omega}^*(\boldsymbol{\gamma}(t))}
\left[\widetilde P^*_{k}(\boldsymbol{\gamma})\right]-\mathbb{E}^{\Omega^*(\boldsymbol{\gamma}(t))}
\left[P^*_{k}(\boldsymbol{\gamma})\right]
\nonumber\\
=&\sum_{\hat{\boldsymbol{\chi}}}
\widetilde \pi^*(\hat{\boldsymbol{\chi}})\mathbb E\left[\widetilde P_k^*(\boldsymbol{\gamma})-P_k^*(\boldsymbol{\gamma})
+\sum_{\hat{\boldsymbol{\chi}}^{'}}(\Pr\{\hat{\boldsymbol{\chi}}^{'}| \boldsymbol{\chi},\widetilde\Omega^*(\boldsymbol{\gamma},\hat{\boldsymbol{\chi}})\}-
\Pr\{\hat{\boldsymbol{\chi}}^{'}| \boldsymbol{\chi},\Omega^*(\boldsymbol{\gamma},\hat{\boldsymbol{\chi}})\})V(\widetilde{\mathbf Q}^{'})\bigg| \hat{\boldsymbol{\chi}}\right],\nonumber
\end{align}
where $\widetilde P_k^*(\boldsymbol{\gamma})$ and $P_k^*(\boldsymbol{\gamma})$ are the power consumptions of BS $k$ given by $\widetilde\Omega^*(\boldsymbol{\gamma}(t),\hat{\boldsymbol{\chi}})$ and $\Omega^*(\boldsymbol{\gamma}(t),\hat{\boldsymbol{\chi}})$, respectively. }
$\widetilde\pi^*(\cdot)$ is the steady state distribution of observed
state $\hat{\boldsymbol{\chi}}$ under the policy $\widetilde \Omega^*(\boldsymbol \gamma(t))$,
$V(\cdot)$ is the potential function of observed state under the
policy $\widetilde{\Omega}^*$.
%From
%$||P_k(\hat{\boldsymbol{\chi}})-\widetilde{P}_k(\hat{\boldsymbol{\chi}})||=O(P_{\text{cct}})+O(\epsilon)$,
%and
%$||R_k(\hat{\boldsymbol{\chi}})-\widetilde{R}_k(\hat{\boldsymbol{\chi}})||=O(P_{\text{cct}})+O(\epsilon)$,
Since \textcolor{black}{ $\widetilde\pi^*(\cdot)$} and
$\{V(\hat{\boldsymbol{\chi}}^{'})-V(\hat{\boldsymbol{\chi}}):\forall
\hat{\boldsymbol{\chi}},\hat{\boldsymbol{\chi}}^{'} \}$ are bounded and $\widetilde P_k^*(\boldsymbol{\gamma})-P^*_k(\boldsymbol{\gamma})=O(P_{\text{cct}})$,
we have
$|\boldsymbol{\phi}_{(k,P)}(\boldsymbol{\gamma})|=O(P_{\text{cct}})+O(\epsilon).$
Similarly, we have
$\big|\boldsymbol{\phi}_{(k,C)}(\boldsymbol{\gamma})|=O(P_{\text{cct}})+O(\epsilon).$ 
Denote \textcolor{black}{$\boldsymbol{\phi}(\boldsymbol{\gamma})=[\boldsymbol{\phi}_{(k,P)}(\boldsymbol{\gamma}) \
\boldsymbol{\phi}_{(k,C)}(\boldsymbol{\gamma})]^T$}. Then, we have $||\boldsymbol{\phi}(\boldsymbol{\gamma})||=O(P_{\text{cct}})+O(\epsilon)$, \textcolor{black}{which implies $||\boldsymbol{\phi}(\boldsymbol{\gamma})||^2=\delta_1+\delta_2$.}
\textcolor{black}{Now, we establish the relationship between the ODEs in \eqref{eq:gamma_ODE} (for general $\epsilon$ and $P_{\text{cct}}$) and \eqref{eq:ODE_fake} (for $\epsilon=0$ and $P_{\text{cct}}=0$) using $\boldsymbol{\phi}(\boldsymbol{\gamma})$: $\dot{\boldsymbol{\gamma}}(t)=f_{\widetilde\Omega^*}(\boldsymbol{\gamma}(t))=f_{ \Omega^*}(\boldsymbol{\gamma}(t))+\boldsymbol{\phi}(\boldsymbol{\gamma}(t))$.} 
Then, we have
\begin{equation}\begin{array}{lll}
\dot{{\red \mathcal L}}(\boldsymbol{\gamma})\triangleq
\frac{\text{d}{\red \mathcal L}}{\text{d}t}&=&\frac{\text{d}{\red \mathcal L}}{\text{d}\boldsymbol{\gamma}}\dot{\boldsymbol{\gamma}}
=\frac{\text{d}{\red \mathcal L}}{\text{d}\boldsymbol{\gamma}}(
f_{\textcolor{black}{ \Omega^*}}(\boldsymbol{\gamma})+\boldsymbol{\phi}(\boldsymbol{\gamma})
) \leq-C_3||\boldsymbol{\gamma}-\boldsymbol{\gamma}_0^*||^2+2
C_2||\boldsymbol{\gamma}-\boldsymbol{\gamma}_0^*||\cdot||\boldsymbol{\phi}(\boldsymbol{\gamma})||\\
&=&-||\boldsymbol{\gamma}-\boldsymbol{\gamma}_0^*||(C_3||\boldsymbol{\gamma}-\boldsymbol{\gamma}_0^*||-
2 C_2||\boldsymbol{\phi}(\boldsymbol{\gamma})|| )
\end{array}\nonumber
\end{equation}
Note that $\dot{{\red \mathcal L}}(\boldsymbol{\gamma})<0$ for all
$\boldsymbol{\gamma}$ s.t.
$(C_3)^2||\boldsymbol{\gamma}-\boldsymbol{\gamma}_0^*||^2\geq
4(C_2)^2||\boldsymbol{\phi}(\boldsymbol{\gamma})||^2=\textcolor{black}{\delta_1+\delta_2}$. As a result, $\boldsymbol{\gamma}^t$ converges almost surely to an
invariant set given by $ \mathcal{S}\triangleq
\left\{\boldsymbol{\gamma}:
||\boldsymbol{\gamma}-\boldsymbol{\gamma}_0^*||^2-\textcolor{black}{\delta_1+\delta_2} \leq0
\right\}$. Furthermore, from $\dot{{\red \mathcal L}}(\boldsymbol{\gamma}^*)=0$, we
have $||\boldsymbol{\gamma}^*-\boldsymbol{\gamma}_0^*||^2-\textcolor{black}{\delta_1+\delta_2}
\leq0$. Therefore, the invariance set is also given by $
\mathcal{S}\triangleq \left\{\boldsymbol{\gamma}:
||\boldsymbol{\gamma}-\boldsymbol{\gamma}^*||^2-\textcolor{black}{\delta_1+\delta_2}
\leq0 \right\}$.

\section*{Appendix F: Proof of Lemma~\ref{lem:action}} \label{app:action}

\textcolor{black}{When $\mathbf{H}=\hat{\mathbf{H}}$, there is no interference under the Pco-MIMO.} Thus, $\frac{
\partial
g_{n}^{\hat{\boldsymbol{\chi}}}(\mathbf{P},\mathbf{R}_n,\mathbf{H})}{\partial
a_k}=0$ in \eqref{eq:gradient}. Therefore, when $P_{\text{cct}}=0$ and $\epsilon=0$ (i.e.,  $\mathbf{H}=\hat{\mathbf{H}}$),  the vector form of the iterations in Algorithm \ref{Alg:gra} becomes
$\mathbf{W}^{t+1}=\mathbf{W}^{t}-\kappa_{v}(t)
\nabla h^{\hat{\boldsymbol{\chi}}}(\mathbf{W}^{t}),$
where  $\mathbf{W}^{t}=\{\mathbf{P}^t,\mathbf{R}^t\}$ denotes the vector of the control actions at frame $t$ and $\nabla h^{\hat{\boldsymbol{\chi}}}(\mathbf{W}^{t})$ denotes the vector of the partial gradients (the first term   in \eqref{eq:gradient}). \textcolor{black}{Note that, when  $\mathbf{H}=\hat{\mathbf{H}}$, $\nabla h^{\hat{\boldsymbol{\chi}}}(\mathbf{W}^{t})$} is deterministic instead of stochastic. Using the standard gradient update argument \textcolor{black}{\cite[Chap. 10]{Borkar:2008}}, $\mathbf{W}^{t}$ tracks the
trajectory of the ODE $\dot{ \mathbf{W}}(t)=-\nabla
h^{\hat{\boldsymbol{\chi}}}( \mathbf{W}(t))$ and $ \mathbf{W}^{t}$
converges to a local minimum   $\hat{\mathbf{W}}^*=\{\hat{\mathbf{P}}^*,\hat{\mathbf{R}}^*\}$ in
$\hat{\mathcal{W}}^*$ \textcolor{black}{as $t\to \infty$}. When $P_{\text{cct}}\neq0$ and $\epsilon\neq0$, the vector form \textcolor{black}{of the iterations in Algorithm \ref{Alg:gra}} can be written as $ \mathbf{W}^{t+1} =  \mathbf{W}^{t} +
\kappa_{v}(t)[ -\nabla
h^{\hat{\boldsymbol{\chi}}}( \mathbf{W}^{t})
-
\boldsymbol{\phi}( \mathbf{W}^{t}) + \mathbf{M}^{t+1}
],$ where
$||\boldsymbol{\phi}( \mathbf{W}^{t})||=O(P_{\text{cct}})+O(\epsilon)$ 
and $ \mathbf{M}^{t+1}$ is a Martingale difference
noise with
$\mathbb{E}[\mathbf{M}^{t+1}|\hat{\mathbf{H}}]=0$.
Following the same argument in the proof of Lemma
\ref{lem:LM_converge}, we can prove Lemma~\ref{lem:action}.

%\bibliographystyle{IEEEtran}
%% argument is your BibTeX string definitions and bibliography database(s)
%\bibliography{IEEEabrv,backhaul}

\begin{thebibliography}{10}
\providecommand{\url}[1]{#1}
\csname url@samestyle\endcsname
\providecommand{\newblock}{\relax}
\providecommand{\bibinfo}[2]{#2}
\providecommand{\BIBentrySTDinterwordspacing}{\spaceskip=0pt\relax}
\providecommand{\BIBentryALTinterwordstretchfactor}{4}
\providecommand{\BIBentryALTinterwordspacing}{\spaceskip=\fontdimen2\font plus
\BIBentryALTinterwordstretchfactor\fontdimen3\font minus
  \fontdimen4\font\relax}
\providecommand{\BIBforeignlanguage}[2]{{%
\expandafter\ifx\csname l@#1\endcsname\relax
\typeout{** WARNING: IEEEtran.bst: No hyphenation pattern has been}%
\typeout{** loaded for the language `#1'. Using the pattern for}%
\typeout{** the default language instead.}%
\else
\language=\csname l@#1\endcsname
\fi
#2}}
\providecommand{\BIBdecl}{\relax}
\BIBdecl

\bibitem{DahroujYu:2010}
H.~Dahrouj and W.~Yu, ``Coordinated beamforming for the multicell multiantenna
  wireless system,'' \emph{{IEEE} Trans. Wireless Commun.}, vol.~9, no.~5, p.
  1748Ð1759, May 2010.

\bibitem{Gesbertadaptation:2007}
D.~Gesbert, S.~{G. Kiani}, A.~Gjendemsj¿, and G.~{E. ¯ien}, ``Adaptation,
  coordination and distributed resource allocation in interference-limited
  wireless networks,'' in \emph{Proc. of Institute of Electrical and
  Electronics Engineers}, vol.~95, no.~12, 2007, pp. 2393--2409.

\bibitem{RenLaDLbf:2006}
T.~Ren and R.~La, ``Downlink beamforming algorithms with inter-cell
  interference in cellular networks,'' \emph{{IEEE} Trans. Wireless Commun.},
  vol.~5, no.~10, p. 2814Ð2823, Oct. 2006.

\bibitem{FoschiniNetworkMIMO:2006}
M.~Karakayali, G.~Foschini, and R.~Valenzuela, ``Network coordination for
  spectrally efficient communications in cellular systems,'' \emph{{IEEE}
  Wireless Commun. Mag.}, vol.~13, no.~4, pp. 56--61, Apr. 2006.

\bibitem{ShamaicoMIMO:2001}
S.~Shamai and B.~Zaidel, ``Enhancing the cellular downlink capacity via
  co-processing at the transmitting end,'' in \emph{Proc. of IEEE Vehicular
  Technology Conference-Spring}, 2001, p. 1745Ð1749.

\bibitem{zhang:coor:2004}
{H. Zhang and H. Dai}, ``{Cochannel interference mitigation and cooperative
  processing in downlink multicell multiuser MIMO networks},'' in \emph{EURASIP
  J. Wireless Commun. Netw.}, {no. 2, pp. 222-235, 2004}.

\bibitem{Gesbert:jsac:2010}
D.~Gesbert, S.~Hanly, H.~Huang, S.~S. Shitz, O.~Simeone, and W.~Yu,
  ``{Multi-cell MIMO cooperative networks: A new look at interference},''
  \emph{{IEEE} J. Sel. Areas Commun.}, vol.~28, pp. 1--29, Dec. 2010.

\bibitem{BjornsonEduardMulticellMIMO:2012}
E.~Bjornson and E.~Jorswieck, ``Optimal resource allocation in coordinated
  multi-cell systems,'' \emph{Foundations and Trends in Communications and
  Information Theory}, vol.~9, no. 2-3, p. 113Ð381, 2012.

\bibitem{Marsch08onmulti-cell}
P.~Marsch and G.~Fettweis, ``On multi-cell cooperative transmission in
  backhaul-constrained cellular systems,'' 2008.

\bibitem{MIMO:int:huang}
C.~Huang and S.~A. Jafar, ``{Degrees of freedom of the MIMO interference
  channel with cooperation and cognition},'' \emph{{IEEE} Trans. Inf. Theory},
  vol.~55, pp. 4211--4220, Sep. 2009.

\bibitem{RandapartialcoMIMO:2011}
R.~Zakhour and D.~Gesbert, ``{Optimized Data Sharing in Multicell {MIMO} with
  Finite Backhaul Capacity},'' vol.~59, no.~12, Dec. 2011, pp. 6102--6111.

\bibitem{SiddarthpartialcoMIMO:2010}
S.~Hari and W.~Yu, ``{Partial zero-forcing precoding for the interference
  channel with partially cooperating transmitters},'' in \emph{Proc. ISIT},
  June 2010, pp. 2283--2287.

\bibitem{Bertsekas:2007}
D.~Bertsekas, \emph{{Dynamic programming and optimal control}}.\hskip 1em plus
  0.5em minus 0.4em\relax Athena Scientific, 2007, vol.~2.

\bibitem{TDD:2006}
A.~Pascual-Iserte, D.~P. Palomar, A.~I. Prez-Neira, and M.~A. Lagunas, ``{A
  robust maximin approach for MIMO communications with partial channel state
  information based on convex optimization},'' \emph{{IEEE} Trans. Signal
  Process.}, vol.~54, pp. 346--360, Jan. 2006.

\bibitem{FDD:2007}
M.~Botros and T.~N. Davidson, ``{Convex conic formulations of robust downlink
  precoder designs with quality of service constraints},'' \emph{{IEEE} J. Sel.
  Areas Signal Process.}, vol.~1, pp. 714--724, Dec. 2007.

\bibitem{Borkaractorcritic:2005}
V.{S.Borkar}, ``An actor-critic algorithm for constrained markov decision
  processes,'' in \emph{Systems Control Lett. 54}, 2005, pp. 207--213.

\bibitem{Meuleau:1999}
N.~Meuleau, K.~E. Kim, L.~P. Kaelbling, and A.~R. Cassandra, ``Solving pomdps
  by searching the space of finite policies,'' in \emph{Proc. of the Fifteenth
  Conf. on Uncertainty in AI}, 1999, pp. 417--426.

\bibitem{LTE:2008}
G.~R1-084026, LTE-Advanced Evaluation Methodology, Oct 2008.

\bibitem{weiyu:2010}
S.~Hari and W.~Yu, ``{Partial zero-forcing precoding for the interference
  channel with partially cooperating transmitters},'' in \emph{Proc. ISIT},
  June 2010.

\bibitem{Spencer:Tsig2004}
Q.~H. Spencer, A.~L. Swindlehurst, and M.~Haardt, ``{Zero-forcing methods for
  downlink spatial multiplexing in multiuser MIMO channels},'' \emph{{IEEE}
  Trans. Signal Process.}, vol.~52, pp. 461--471, Feb. 2004.

\bibitem{MIMO:int:jafar}
S.~A. Jafar and M.~Fakhereddin, ``{Degrees of freedom for the MIMO interference
  channel},'' \emph{{IEEE} Trans. Inf. Theory}, vol.~53, pp. 2637--2642, Jul.
  2007.

\bibitem{Ross:2003}
S.~M. Ross, \emph{{Introduction to probability models}}.\hskip 1em plus 0.5em
  minus 0.4em\relax 8th edition, Amsterdam : Academic Press, 2003.

\bibitem{Tse04fundamentals}
D.~Tse and P.~Viswanath, \emph{Fundamentals of Wireless Communications}.\hskip
  1em plus 0.5em minus 0.4em\relax Cambridge University Press, 2004.

\bibitem{Berry:2002}
R.~A. Berry and R.~Gallager, ``Communication over fading channels with delay
  constraints,'' \emph{{IEEE} Trans. Inf. Theory}, vol.~48, pp. 1135--1148, May
  2002.

\bibitem{Thesis:Salodkar}
N.~Salodkar, ``{Online Algorithms for Delay Constrained Scheduling over a
  Fading Channel},'' Ph.D. dissertation, Indian Institute of Technology, {May}
  {2008}.

\bibitem{Borkar:2008}
V.~S. Borkar, \emph{{Stochastic Approximation: A Dynamical Systems Viewpoint
  }}.\hskip 1em plus 0.5em minus 0.4em\relax Cambridge University Press, 2008.

\bibitem{Kenneth:2003}
K.~E. Train, \emph{{Discrete choice methods with simulation}}.\hskip 1em plus
  0.5em minus 0.4em\relax Cambridge, UK: Cambridge University Press, 2003.

\bibitem{Kanwal:1998}
R.~P. Kanwal, \emph{{Generalized functions: theory and technique}}.\hskip 1em
  plus 0.5em minus 0.4em\relax 2nd ed. Boston, MA: Birkhauser, 1998.

\bibitem{Vincent:MIMO}
V.~K.~N. Lau and Y.~Chen, ``{Delay-optimal power and precoder adaptation for
  multi-stream MIMO systems},'' \emph{{IEEE} Trans. Wireless Commun.}, vol.~8,
  pp. 3104--3111, Jun. 2009.

\bibitem{borkar:probability:1995}
V.~S. Borkar, \emph{{Probability theory: an advanced course}}.\hskip 1em plus
  0.5em minus 0.4em\relax Springer Verlag, New York, 1995.

\bibitem{Borkar::convex:2001}
------, ``{Convex analytic methods in markov decision processes},''
  \emph{{Feinberg, A. Schwartz (Eds.) Handbook of Markov Decision Processes,
  Kluwer Academic Publishers}}, pp. 347--375, 2001.

\bibitem{LP:infinite:1987}
E.~J. Anderson and P.~Nash, \emph{{Linear programming in infinite dimensional
  spaces}}.\hskip 1em plus 0.5em minus 0.4em\relax John Wiley, Chichester,
  1987.

\bibitem{Lyapunov:1998}
M.~Corless and L.~Glielmo, ``{New converse Lyapunov theorems and related
  results on exponential stability},'' \emph{{Math. Control Signals Systems}},
  pp. 79--100, 1998.

\end{thebibliography}


\end{document}